\documentclass[pre,twocolumn,showpacs,floatfix, superscriptaddress]{revtex4-1}
\usepackage{amsmath,amssymb,amsfonts,amsthm}
\usepackage{mathcomp}
\usepackage{comment}
\usepackage{graphicx,subfigure}
\usepackage{txfonts}
\usepackage{epstopdf}
\usepackage[title]{appendix}
\usepackage{color}
\usepackage{tikz}
\usetikzlibrary{calc}
\usepackage{quantikz}
\usetikzlibrary{positioning}
\usepackage{tikz-cd}
\newcommand{\bnum}[1]{\textcolor{blue!70!black}{#1}}
\newcommand{\gnum}[1]{\textcolor{green!60!black}{#1}}
\newcommand{\rnum}[1]{\textcolor{red!60!black}{#1}}
\usepackage{soul}
\usepackage[normalem]{ulem}
\usepackage{physics}
\usepackage{bbold}
\usepackage{footnote}
\usetikzlibrary{positioning,arrows.meta}
\usepackage{graphicx}
\usepackage{algorithm}
\usepackage{algorithmic}
\usepackage{nicematrix}

\usepackage{xr-hyper} 
\usepackage[unicode]{hyperref} 
\hypersetup{
	unicode=true,                                           % non-Latin characters in AcrobatÃÂ¢ÃÂÃÂs bookmarks
	plainpages=false,
	pdffitwindow=true,                                      % page fit to window when opened
	colorlinks=true,                                        % false: boxed links; true: colored links
	linkcolor=blue,                                        % color of internal links
	citecolor=blue,                                         % color of links to bibliography
	filecolor=blue,                                        % color of file links
	urlcolor=blue                                          % color of external links
}
\usetikzlibrary{arrows.meta,positioning,calc,shapes.geometric}

\usepackage{xcolor}
\definecolor{lightgray}{gray}{0.97}

\usepackage[many]{tcolorbox}    	% for COLORED BOXES (tikz and xcolor included)
\tcbset{
    rounded corners,
    colback = lightgray,
    before skip = 0.25cm,    % add extra space before the box
    after skip = 0.25cm,      % add extra space after the box
    left=0mm,
    right=0mm,
    top=0mm,
    bottom=0mm
}                           % setting global options for tcolorbox
\newtcolorbox{boxA}{
    boxrule = 1.0pt,
    colframe = black, % frame color
    arc = 3pt   % corners roundness
}

\newtheorem{lemma}{Lemma}

\begin{document}
\title{Graph automorphisms to obtain Clifford symmetries in open and closed qudit models}

\author{Charlie Nation}
\email{c.nation2@exeter.ac.uk}
\affiliation{Department of Physics and Astronomy, University of Exeter, Stocker Road, Exeter EX4 4QL, United Kingdom}
\affiliation{QuAOS collaboration}

\author{Rick P. A. Simon}
\affiliation{Department of Physics and Astronomy, University of Exeter, Stocker Road, Exeter EX4 4QL, United Kingdom}
\affiliation{QuAOS collaboration}

\author{Shreya Banerjee}
\affiliation{Department of Physics and Astronomy, University of Exeter, Stocker Road, Exeter EX4 4QL, United Kingdom}
\affiliation{QuAOS collaboration}

\author{Francesco Martini}
\affiliation{QuAOS collaboration}

\author{Alessandro Ricottone}
\affiliation{QuAOS collaboration}

\author{Federico Cerisola}
\affiliation{Department of Physics and Astronomy, University of Exeter, Stocker Road, Exeter EX4 4QL, United Kingdom}
\affiliation{QuAOS collaboration}

\author{Luca Dellantonio}
\email{l.dellantonio@exeter.ac.uk}
\affiliation{Department of Physics and Astronomy, University of Exeter, Stocker Road, Exeter EX4 4QL, United Kingdom}
\affiliation{QuAOS collaboration}

\date {\today}
\begin{abstract}
    In the recent article [\href{https://arxiv.org/abs/2605.18966}{arXiv:2605.18966}], we demonstrated that finding Clifford symmetries can be mapped to a Graph Automorphism (GA) problem. Here, we provide an algorithm to obtain such symmetries on general qudit systems, that works on the principle of encoding Clifford invariants of a Hamiltonian onto properties of a graph. Labelling Hamiltonian terms as vertices, a permutation of such vertices that respects the Clifford invariants (a GA) is both a valid Clifford, and a symmetry up to phase correction checks. We test this on multiple physical models and discuss the scaling with respect to the number of qudits and Pauli strings, as well as various strategies for optimisation in different regimes. We further show that the graph automorphism representation of Clifford symmetries can be expanded to open quantum systems.
\end{abstract}
\maketitle

\section{Introduction}

The study of symmetry in quantum systems is vital to almost all phenomena in physics. Symmetries reduce complexity, and their knowledge is vital in determining properties such as the ground state energy \cite{Gibbs2025,Yoo_symmetry_2025}, the dynamics \cite{Tran_faster_2021}, and thermodynamics \cite{Vidmar_2016, DAlessio2016} of a model. The breaking of a symmetry can further yield exotic phenomena, such as magnetic domains and superconductivity \cite{sachdev2023quantum}, and emergent symmetries can lead to localization effects \cite{AnushyaLIOM2015, Shtanko2025}. 

Symmetries can split the Hilbert space of a quantum system into disjoint sectors \cite{Gunderman2024}, and in more constrained settings this yields Hilbert-space fragmentation, where the dynamics breaks into disconnected subspaces \cite{Sala2020HSF,MoudgalyaMotrunich2022HSF}. Knowledge of such symmetries is valuable not only for analytical understanding but also for numerical methods, where symmetry-adapted bases and block-diagonalization can dramatically reduce quantum and classical computational cost \cite{SchmitzJohri2020SymmetryED,Alvarez2012SU2DMRG,MoudgalyaMotrunich2023NumericalSymmetries, paulson_2021_simulating}. Additionally, hidden symmetries in open systems can be exploited for noise robust long-range entanglement \cite{DuttaCoherence2020}. 

Symmetries are usually discovered from inspection of the Hamiltonian. However, in some cases, they are far from obvious \cite{Tasaki1992, Yangpairing1989, SZRIFTGISER2021108076, Lootens2023, Pauli1926, CarigliaHidden2014}.
Algorithmic approaches to obtain symmetries in quantum systems are limited to Pauli symmetries \cite{Gunderman2024, bravyi2017tapering}, or become computationally expensive at increasing system sizes \cite{LaBorde_quantum_2022, Moudgalya_symmetries_2023}. 
In this work, we are interested in the larger class of Clifford symmetries, those efficiently implementable on a classical computer \cite{gottesman_class_1996, gottesman_surviving_nodate, aaronson_improved_2004, Chapman2020characterizationof, JonesIntegrable2022}. As of today, there is no algorithm that can find them. Here, we fill this gap.

In Ref.~\cite{main}, we showed that finding Clifford symmetries can be efficiently represented as a Graph Automorphism (GA) problem. In this work, we present the algorithm to solve this problem in detail, exploiting available invariants and heuristics. We further analyse these symmetries on multiple systems, extending the Hamiltonian qubit representation of Ref.~\cite{main} to Liouvillian systems on qudits. In other words, given any finite open or closed, arbitrary spin system \cite{simon2026error}, our algorithm will find its Clifford symmetries. Furthermore, at every step we use the symplectic representation of qudit operators, which enables efficient classical representation throughout.

This article is arranged as follows. We introduce the problem in Sec.~\ref{app:clifford_summary}, detailing the representation of Pauli Hamiltonians by a tableau, and operations of Clifford circuits via symplectic matrices. In Sec.~\ref{app:clifford_symm} we discuss how a Clifford symmetry can be described in such a representation as either a symplectic map, or a permutation of Paulis, and extend the approach to open quantum systems. This permits a graph representation of the Hamiltonian (Sec.~\ref{app:building_graph}), whereby a symmetry can be expressed as a GA. 
%After discussing the role of phases in Sec.~\ref{app:phase_correction} 
We then show in Sec.~\ref{app:algorithm} how to encode Clifford invariants into the GA problem, describing an efficient algorithm to compute the symmetry. Sec.~\ref{app:numerical_results} includes several examples, highlighting algorithmic improvements that are obtained from e.g. colour-coding of the graph. We conclude with a discussion of the potential applications of this method in Sec.~\ref{app:conclusion}. Python code for our approach is provided in \cite{sympleq}.

\section{Clifford formalism and symplectic spaces}\label{app:clifford_summary}

We consider an $n$ qudit system, and define the Pauli operators acting on the $\alpha$-th qudit with prime \footnote{
The restriction to prime $d$ ensures that the Pauli labels form a vector space over the field $\mathbb{F}_d$, allowing the use of Gaussian elimination and symplectic linear algebra. For composite $d$ the labels instead live over the ring $\mathbb{Z}_d$, where nonzero elements need not be invertible. In many cases composite-dimensional systems may also be decomposed into smaller prime-dimensional subsystems.
}
dimension $d$ via \cite{Hostens_stabilizer_2005}

\begin{equation}\label{eq:Pauli_def}
    \hat{X}_{\alpha}|\vec{k}\rangle 
    =  
    |\vec{k} + \vec{\alpha} \rangle 
    ,\quad
    \hat{Z}_\alpha|\vec{k}\rangle 
    = 
    e^{2 \pi i \frac{k_\alpha}{d}} |\vec{k}\rangle
    .
\end{equation}
Here, $\vec{k} \in \mathbb{Z}^{n}_d$, with $\mathbb{Z}_d$ the field of (prime) integers $\mod{d}$ \cite{dehaene_clifford_2003, Hostens_stabilizer_2005, galois_qudits_2026}, $\vec{\alpha}$ is the unit vector of length $n$ such that $k_\alpha = \vec{k} \cdot \vec{\alpha} \in \mathbb{Z}_d$ is the $\alpha$-th element of $\vec{k}$. We will assume a fixed dimension $d$ for all qudits in the main text, and in Appendix~\ref{app:mixed_d} describe the generalisation to mixed qudit dimensions. The starting point for the Clifford formalism is the representation of a Pauli string as a vector $\vec{p} \in \mathbb{Z}_d^{2n}$. In our convention, $\vec{p}$ is a row vector, and the first (last) $n$ entries are the powers of the Pauli $\hat{X}_\alpha$ ($\hat{Z}_\alpha$) operators on each site. We thus have $\vec{p} = [\vec{x}_i \, | \, \vec{z}_i]$ with $\vec{x},\vec{z} \in \mathbb{Z}^{n}_d$. We write an arbitrary Pauli string in operator form as $\hat{P}(\vec{p})$, such that
\begin{equation}
    \hat P(\vec p)
    =
    \bigotimes_{\alpha=1}^n \hat X_\alpha^{x_\alpha}\hat Z_\alpha^{z_\alpha}.
\end{equation}
A Hamiltonian on $n$-qudits can then be expressed as a weighted sum over Paulis $\hat{P}$ via
\begin{equation}\label{eq:pauli_sum}
    \hat{H} 
    = 
    \sum_{i=1}^M c_i e^{{\rm i} \pi \frac{\eta_i}{d}} \hat{P}(\vec{p}_i)
    \quad
    \text{and}
    \quad
    \underline{H}
    =
    \begingroup
    \setlength{\arraycolsep}{3pt} % <- try 2pt, 1.5pt, ...
    \left[
    \begin{array}{@{}c@{}|@{}c@{}|@{}c@{}}
        \smash{\overbrace{\vphantom{\vec{p}_M}\hphantom{c_M}}
            ^{\vphantom{\vec{\eta}\underline{p}}\smash{\vec{c}}}}
        &
        \smash{\overbrace{\vphantom{\vec{p}_M}\hphantom{\vec{p}_M}}
            ^{\vphantom{\vec{\eta}\underline{p}}\smash{\underline{p}}}}
        &
        \smash{\overbrace{\vphantom{\vec{p}_M}\hphantom{\eta_M}}
            ^{\vphantom{\vec{\eta}\underline{p}}\smash{\vec{\eta}}}}
        \\[-3ex]
        c_1 & \vec{p}_1 & \eta_1 \\
        \vdots & \vdots & \vdots \\
        c_M & \vec{p}_M & \eta_M
    \end{array}
    \right]
    \endgroup
    .
\end{equation}

The tableau $\underline{H} \cong{\hat{H}}$ is thus fully defined by a vector of coefficients $\vec{c}$, an exponent matrix $\underline{p}$ (where each row describes a Pauli string $\vec{p}_i$), and a vector of phases $\vec{\eta}$, where $\eta_i\in\mathbb Z_{2d}$ are tracked modulo $2d$ \cite{main, Hostens_stabilizer_2005}. 
%An example tableau for the Ising chain on 3 qubits is shown in Fig.~\ref{fig:ising_eg}(a). 
In Sec.~\ref{sec:pauli_lindblad} we show how this representation is generalised to open quantum systems. This requires doubling the number of qudits, which is classically efficient with the tableau representation employed here.

The power of the Clifford formalism goes beyond an efficient representation of Hamiltonians and Paulis. Rather, it lies in the efficient operations that we may perform on them without leaving this representation. In Hilbert space, operations require unitaries of size $2^n \times 2^n$, whereas in the Clifford representation, a subset of these operations can be performed by symplectics of size $2n \times 2n$. This allows to define the classically efficient subgroup of possible quantum operations, and enables a wealth of applications, from quantum error correcting codes \cite{gottesman_fault-tolerant_1999, gottesman_class_1996, gottesman_surviving_nodate, rengaswamy_logical_2020}, benchmarking and synthesis of quantum circuits \cite{Polloreno2025theoryofdirect, Mitrarai_quadratic_2022}, `nonstabilizerness' as a measure of quantum resources \cite{TurkeshiTirritoSierant2025MagicSpreading, Liu_magic_2022, Bejan_Dynamical_Magic_2024}, and pre-computation for both quantum and classical algorithms \cite{Qassim2019clifford, Huggins2024acceleratingquantum, Khosla_Clifford_2025}.

To represent these efficient operations, we first denote the symplectic form $\underline{\Omega}$ \footnote{
Notice that the definition of $\underline{\Omega}$ in Eq.~\eqref{eq:symplectic_form} is equivalent to the one given in Ref.~\cite{main} for $d=2$.
} and its lowered-reduced version $\underline{\Omega}_{\rm l}$,
\begin{equation}
\label{eq:symplectic_form}
\underline{\Omega} 
=
\begin{bmatrix} 
0 & \underline{\mathbb 1}_n 
\\[2pt] 
-\underline{\mathbb 1}_n & 0 
\end{bmatrix}
\quad
\text{and}
\quad
\underline{\Omega}_{\rm l} 
=
\begin{bmatrix} 
0 & 0 
\\[2pt] 
\underline{\mathbb 1}_n & 0 
\end{bmatrix}
,
\end{equation}
with $\underline{\mathbb 1}_n$ denoting the $n \times n$ identity matrix. $\underline{\Omega}$ defines the symplectic (inner) product of two vectors $\vec{p}_i, \, \vec{p}_j \in\mathbb{Z}_d^{2n}$
\begin{equation}\label{eq:symplectic_product_def}
\langle \vec{p}_i,\, \vec{p}_j\rangle \;=\; \vec{p}_i \underline{\Omega}\, \vec{p}_j^\top \;\in\; \mathbb{Z}_d,
\end{equation}
which can be used to determine the commutation of Pauli operators:
\begin{equation}
\hat{P}(\vec{p}_i)\,\hat{P}(\vec{p}_j)
= 
e^{
\frac{-
2 i \pi
}{
d
}
\langle \vec{p}_i,\vec{p}_j\rangle
}
\hat{P}(\vec{p}_j)\, \hat{P}(\vec{p}_i)
.
\end{equation}
As required, for commuting Pauli string operators $P(\vec{p}_i)$ and $P(\vec{p}_j)$ the symplectic product $\langle \vec{p}_i, \vec{p}_j\rangle$ vanishes.

Clifford gates act upon Paulis via right-side \footnote{We note this convention infers simply the interpretation of Paulis as rows rather than columns of the matrix $\underline{p}$} multiplication of their corresponding symplectic matrix $\underline{F} \in \mathbf{Sp}(2n, \mathbb{Z}_d)$ \cite{dehaene_clifford_2003, Hostens_stabilizer_2005}, defined such that $ \underline{F}\, \underline{\Omega} \,  \underline{F}^\top = \underline{\Omega}$. For instance, a Pauli $\vec{p}_i$ can be mapped to another Pauli $\vec{p}_j$ via
\begin{equation}\label{eq:symplectic_action}
    \vec{p}_j = \vec{p}_i\underline{F}.
\end{equation}
Simultaneously, the phase $\eta_i$ of the Pauli $\hat{P}(\vec{p}_i)$ is updated upon action of $\underline{F} $ by \cite{Hostens_stabilizer_2005}
\begin{equation}\label{eq:phase_update}
        \eta_i \mapsto \eta_i + [\vec{\phi} + \textrm{diag}(\mathcal{R}_1(\underline{F}))] \cdot \vec{p}_i + \vec{p}_i \mathcal{R}_2(\underline{F}) \vec{p}_i^\top  \, \, \, \textrm{mod}\, 2d ,
\end{equation}
where $\underline{\Omega}_{\rm l}$ is in Eq.~\eqref{eq:symplectic_form}.
Furthermore, we have defined $
\mathcal{R}_1(\underline{F}) 
= 
(\underline{F}\,\underline{\Omega}_{\rm l} \, \underline{F}^\top)
$, and $
\mathcal{R}_2(\underline{F}) 
= 
2\mathcal{P}_{\textrm{upps}}(\mathcal{R}_1(\underline{F}) 
+ 
\mathcal{P}_{\textrm{diag}}(\mathcal{R}_1(\underline{F}))
$ with $
\mathcal{P}_{\textrm{upps}}(\underline{A})
$ the upper right triangle of a matrix $\underline{A}$ and $\mathcal{P}_{\textrm{diag}}(\underline{A})$
the matrix of diagonals of $\underline{A}$. Finally, $\vec{\phi}$ is a phase vector associated with the Clifford gate, such that a gate $\hat{G}$ is defined in the tableau representation via the pair $\hat{G} \cong (\underline{F}, \vec{\phi})$. 

Usually, Clifford gates are identified only by their symplectic matrices $\underline{F}$, yet they can have different phase vectors $\vec{\phi}$, resulting in different operator representations in Hilbert space. In this work, phases are crucial to identify symmetries of the Hamiltonian, and as such are carefully tracked. The action of a phase $\vec{\phi}$ in Eq.~\eqref{eq:phase_update} can be seen as a Pauli gate $\hat{G}_{\vec{p}}$ that acts on an operator $\hat{O}$, yields $\hat{G}_{\vec{p}}^{\dagger}\hat{O}\hat{G}_{\vec{p}} = \hat{P}(\vec{p})^\dagger \hat{O} \hat{P}(\vec{p})$. The action of a Pauli gate $\hat{G}_{\vec{p}}$ on a Pauli (Pauli sum) changes only the (relative and/or absolute) phase. The symplectic is thus the identity, and a Pauli gate is represented by $\hat{G}_{\vec{p}} \cong (\underline{\mathbb{1}}_{2n}, 2 \vec{p}\underline{\Omega})$ \cite{Hostens_stabilizer_2005}. Conjugation of a Pauli gate with any other gate thus yields an identical symplectic and altered phase vector. This will be exploited below in Sec.~\ref{app:phase_correction} for correcting the phase of a candidate symmetry.

The action of $\underline{F}$ in Eq.~\eqref{eq:symplectic_action} may be understood intuitively, as the rows of $\underline{F}$ are the images of their respective Pauli component. That is, in the Pauli basis, the first row is the image of $x_1$, the second the image of $x_2$, the $(n+1)^{\rm th}$ is the image of $z_1$, and so on. Notably, this fully defines the action on any symplectic basis $\{\vec{u}_i, \, \vec{v}_i\}_{i \in [1,\, n]}$, defined such that $\langle \vec{u}_i, \, \vec{v}_j\rangle = \delta_{ij}$ and $\langle \vec{u}_i,\, \vec{u}_j\rangle = \langle \vec{v}_i,\, \vec{v}_j \rangle = 0$. 

The action of a Clifford has a few properties which we will exploit in the algorithm below. The first is that the Clifford group normalises the Pauli group, i.e., a Clifford maps one Pauli to another. The second is that it leaves the magnitude of the coefficients $c_i$ in Eq.~\eqref{eq:pauli_sum} unchanged, though may alter phases [see Eq.~\eqref{eq:phase_update}] in discrete amounts as tracked by $\eta_i$. This property greatly limits the number of possible automorphisms that we need to check when searching for symmetries (see Appendix~\ref{app:search_scheme}), as we do not need to search for mappings between Paulis with different coefficients. The third property is that the symplectic product between any two Paulis is left invariant under the action of a Clifford:
\begin{equation}\label{eq:adj_is_inv}
    \langle \vec{p}_i, \vec{p}_j\rangle = \langle \vec{p}_i\underline{F}, \vec{p}_j\underline{F}\rangle.
\end{equation}
Finally, as a consequence of linearity, a Clifford leaves the dependency structure of Paulis invariant. That is, given the linear combination $\vec{p}_i = \vec{p}_1 + \vec{p}_2$ (equiv. $\hat{P}(\vec{p}) \propto \hat{P}(\vec{p}_1)\hat{P}(\vec{p}_2)$), applying a Clifford $\underline{F}$ to $\vec{p}_{i} \mapsto \vec{p}_{i'} = \vec{p}_i \underline{F} = \vec{p}_{1'} + \vec{p}_{2'}$ is equivalent to applying it to the basis vectors $\vec{p}_{1'} = \vec{p}_1 \underline{F}$  and $\vec{p}_{2'} = \vec{p}_2 \underline{F}$ first, and summing them. In operator form, $\hat{G}^{\dagger} \hat{P}(\vec{p}) \hat{G} = \hat{G}^{\dagger} \hat{P}(\vec{p}_{1}) \hat{G} \hat{G}^{\dagger} \hat{P}(\vec{p}_{2}) \hat{G}$.

%These properties impose limitations on the set of Hamiltonians reachable via symplectic transformations such as in Eq.~\eqref{eq:symplectic_action}. For example, if $\underline{F} : \vec{p}_i \to \vec{p}_j$, the same symplectic can only map $\vec{p}_k \to \vec{p}_l$ if $\langle \vec{p}_i, \vec{p}_k\rangle = \langle \vec{p}_j, \vec{p}_l\rangle$. The pairwise set of symplectic products of each Pauli $\langle \vec{p}_i, \vec{p}_j\rangle$ is thus invariant under the action of a Clifford, alongside the coefficient vector $\vec{c}$. 

Another important property is that we can define a symplectic map $\underline{F}$ from a set of Paulis $\underline{p}$ to another set $\underline{p}'$ by splitting each set into basis and dependent components. We define elements of the basis set $\underline{p}_{\rm b}$ as $\vec{p}_i^{\rm (b)}$, and note that the number of basis elements $n_{\rm b}$ cannot exceed $2n$. The dependent set elements are then $\vec{p}_j^{\rm (d)} = \sum_{i}^{n_{\rm b}} m_{ij} \vec{p}_i^{\rm (b)}$, where $m_{ij} \in \mathbb{Z}_d$. Then, if we define a target Hamiltonian $\underline{H}'$ with tableau $\underline{p}'$, and $
\langle \vec{p}_i^{\rm (b)}, \vec{p}_{j}^{\rm (b)} \rangle 
= 
\langle \vec{p}_{i'}^{\rm (b)}, \vec{p}_{j'}^{\rm (b)} \rangle
$ for all $i,j$, the symplectic map $\underline{F}_{\underline{H} \mapsto \underline{H}'}$ from $\underline{H} \mapsto \underline{H}'$ is given by 
\begin{equation}\label{eq:basis_map}
    \underline{F}_{\underline{H} \mapsto \underline{H}'} = \underline{p}^{-1}_{\rm b} \, \underline{p}_{\rm b}',
\end{equation}
where the suffix ``b'' indicates that the tableaus only contain the $n_b \leq 2n$ basis elements.
If the Hamiltonian is not full rank ($n_b < 2n$), either the basis can be completed, or the algorithm in Ref.~\cite{rengaswamy_logical_2020} can be used. In this case, the obtained $\underline{F}$ is not unique.

Finally, we note that the formalism can be extended to mixed-dimensional systems by embedding the phase bookkeeping in a common modulus $D=\mathrm{lcm}(d_1,d_2,\ldots,d_n)$ \cite{dongre2026mixed}. One may also exploit the Chinese-remainder decomposition to further separate coprime factors of a composite dimension~\cite{sarkar2024qudit}. To keep the presentation focused, we restrict the main text to homogeneous systems with $d_i=d$, and discuss the mixed-dimensional extension in Appendix~\ref{app:mixed_d}.

\section{Clifford Symmetries}\label{app:clifford_symm}

For a given Hamiltonian, we define a symmetry $\hat{S}$ as any non-trivial gate (that is, not the identity) that maps the Hamiltonian to itself.
This definition is equivalent to $[\hat{H},\hat{S}] = 0$.
Ref.~\cite{main} shows that, in the Clifford formalism, a symmetry must take the form of a row permutation $\underline{\Pi}$ on the Hamiltonian tableau $\underline{H}$ via
\begin{equation} \label{eq:symmetry_definition_tableau}
    \underline{H}' = \underline{\Pi}\, \underline{H}.
\end{equation}
Since the sum of the rows defines the Hamiltonian [Eq.~\eqref{eq:pauli_sum}], the order of the tableaus $\underline{H}$ and $\underline{H}'$ is inconsequential and they represent the same Hamiltonian operator.

As indicated in Eq.~\eqref{eq:symplectic_action}, a Clifford with symplectic action $\underline{F}\in \mathrm{Sp}(2n,d)$ sends $
\underline{p}\ \longmapsto\ \underline{p}\,\underline{F}
$, leaving the coefficient vector $\vec{c}$ in Eq.~\eqref{eq:pauli_sum} unchanged and updating the phase vector $\vec{\eta}$ according to Eq.~\eqref{eq:phase_update}.
Therefore, for $\underline{F}$ to be a symmetry of the Hamiltonian, it must act on the exponent matrix $\underline{p}$ as
\begin{equation}\label{eq:PiPF_intro}
    \underline{\Pi}\,\underline{p} \;=\; \underline{p}\,\underline{F},
\end{equation}
where $\underline{\Pi}$ permutes the rows of $\underline{p}$.
However, since a row reshuffling of the exponent matrix does not necessarily correspond to a compatible relabelling in the phase vector $\vec{\eta}$ and since $\underline{F}$ cannot change $\vec{c}$, not all $\underline{F}$ that satisfy Eq.~\eqref{eq:PiPF_intro} also satisfy Eq.~\eqref{eq:symmetry_definition_tableau}.
The goal of this paper consists in providing an algorithm that searches the space of permutations for those which respect the necessary Clifford invariants and can be phase corrected with a Pauli gate $\hat{G}_{\vec{p}}$ (see Sec.~\ref{app:clifford_summary}).
%The question then becomes if we can find a permutation $\underline{\Pi}$ which respects the necessary Clifford invariants? Further, can the obtained symplectic symmetry $\underline{S}$ be phase corrected with a Pauli gate (symplectic $\underline{I}_{2n}$, leaving $\underline{p}$ unaltered) such that the phase vector $\vec{\eta}$ also remains unaltered? 
If we find such a permutation, we have obtained a Clifford symmetry of the Hamiltonian.

\subsection{Extension to open quantum systems}\label{sec:pauli_lindblad}

The same Pauli-string formalism can be extended from Hamiltonians to a broad class of open-system generators. We consider a Gorini–Kossakowski–Sudarshan–Lindblad (GKSL) master equation \cite{GoriniCompletely1976, Lindblad1976} on $n$ qudits of prime local dimension $d$,
\begin{equation}\label{eq:lindblad_master_eq}
    \frac{d\hat{\rho}}{dt}
    =
    \mathcal{L}(\hat{\rho})
    =
    -\frac{\rm i}{\hbar}
    \left[
    \hat{H},\hat{\rho}
    \right]
    +
    \sum_{\beta=1}^{M_J}
    \left(
        \hat{L}_\beta \hat{\rho}\hat{L}_\beta^\dagger
        -
        \frac{1}{2}
            \left\{
            \hat{L}_\beta^\dagger \hat{L}_\beta,\hat{\rho}
            \right\}
    \right).
\end{equation}
For now we focus on the Pauli-GKSL case in which the jump operators are proportional to Pauli strings,
\begin{equation}\label{eq:jump_op_Pauli}
    \hat{L}_\beta
    =
    \sqrt{\gamma_\beta}\,
    e^{{\rm i}\pi \frac{\xi_\beta}{d}}
    \hat{P}(\vec{q}_\beta),
    \qquad
    \vec{q}_\beta \in \mathbb{Z}_d^{2n}.
\end{equation}
Since $\hat{P}(\vec{q}_\beta)$ is unitary, the overall phase $e^{i\pi \xi_\beta/d}$ cancels in
$\hat{L}_\beta \hat{\rho}\hat{L}_\beta^\dagger$, and $\hat{L}_\beta^\dagger \hat{L}_\beta = \gamma_\beta \hat{\mathbb 1}$. Thus, each dissipative term in the right-hand-side of Eq.~\eqref{eq:lindblad_master_eq} can be rewritten as
\begin{equation}\label{eq:pauli_dissipator}
    \hat{L}_\beta \hat{\rho}\hat{L}_\beta^\dagger
    -
    \frac{1}{2}
        \left\{
        \hat{L}_\beta^\dagger \hat{L}_\beta,\hat{\rho}
        \right\}
    =
    \gamma_\beta
    \left(
        \hat{P}(\vec{q}_\beta)\hat{\rho}\hat{P}(\vec{q}_\beta)^\dagger
        -
        \hat{\rho}
    \right).
\end{equation}

To represent $\mathcal{L}$ in the same spirit as Eq.~\eqref{eq:pauli_sum}, we introduce Pauli superoperators labelled by a pair of Pauli vectors,
\begin{equation*}
    \vec{s}
    =
    (\vec{p}^{\rm \, L},\vec{p}^{\rm \, R})
    \in
    \mathbb{Z}_d^{4n},
\end{equation*}
and define
\begin{equation*}
    \hat{\mathcal P}(\vec{s})[\hat{\rho}]
    =
    \hat{P}(\vec{p}^{\rm \, L})\,
    \hat{\rho}\,
    \hat{P}(\vec{p}^{\rm \, R})^\dagger.
\end{equation*}
In a vectorized (Liouville-space) representation \cite{gilchrist2011vectorization, WeimerSimulation2021} this becomes a matrix acting on $|\hat{\rho}\rangle\!\rangle$, with the precise transpose/conjugation convention depending on the chosen vectorization convention. One possible choice is
$|\hat{A}\hat{\rho} \hat{B}^\dagger\rangle \! \rangle = (\hat{A}\otimes \hat{B}^*)\,|\hat{\rho} \rangle \! \rangle$, with $|\hat{\rho} \rangle \! \rangle$ being the vectorised representation of $\hat{\rho}$ in Liouville space by column stacking $\hat{\rho} = \sum_{ij} \rho_{ij}|i\rangle \langle j| \mapsto |\hat{\rho}\rangle \! \rangle = \sum_{ij} \rho_{ij}|i, j\rangle \!\rangle$. The expectation value may then be written as $\mathrm{Tr}[\hat{O} \hat{\rho}] = \langle \!\langle \mathbb{1}| \hat{O} \hat{\rho} \rangle\! \rangle =  \langle \!\langle \mathbb{1}| \hat{O} \otimes \mathbb{1} | \hat{\rho} \rangle\! \rangle $. We then have,
\begin{equation*}
    \hat{\mathcal P}(\vec{s})|\rho \rangle \!\rangle 
    = 
    \left( \hat P (\vec p^{\rm \, L}) \otimes \hat P(\vec p^{\rm \, R})^* \right) 
    |\rho \rangle \! \rangle .
\end{equation*}
From the definitions in Eq.~\eqref{eq:Pauli_def}, complex conjugation sends
\begin{equation*}
    \hat P([\vec{x}\,|\,\vec{z}])^*
    =
    \hat P([\vec{x}\,|\,-\vec{z}]).
\end{equation*}
Thus the right sector is conjugated in the Liouville-space representation. If
$\vec p^{\rm \, L}=[\vec x^{\rm \, L}\,|\,\vec z^{\rm \, L}]$ and
$\vec p^{\rm \, R}=[\vec x^{\rm \, R}\,|\,\vec z^{\rm \, R}]$, then symplectic products between superoperator labels are
\begin{equation*}
    \langle \vec{s}_i,\vec{s}_j\rangle_{\mathcal L}
    =
    \langle \vec p_i^{\rm \, L},\vec p_j^{\rm \, L}\rangle
    -
    \langle \vec p_i^{\rm \, R},\vec p_j^{\rm \, R}\rangle .
\end{equation*}
Equivalently, one may flip the sign of the right-sector $Z$ exponents before applying the ordinary doubled symplectic form. Accordingly, $\hat{\mathcal P}(\vec{s})$ may be treated as a Pauli string like object on $2n$ qudits, with the same symplectic/Clifford machinery applied using this Liouville-space pairing.

A Pauli-GKSL generator can thus be written as a weighted sum of Pauli superoperators,
\vspace{0.1cm}
\begin{align}\label{eq:pauli_lindblad_sum}
    \hat{\mathcal{L}}
    &=
    \sum_{m=1}^{M_{\mathcal L}}
    \lambda_m
    e^{{\rm i} \pi \frac{\mu_m}{d}}
    \hat{\mathcal P}(\vec{s}_m),
    \; \;
    \underline{\mathcal L}
    &=
    \begingroup
    \left[
    \begin{array}{@{}c@{}|@{}c@{}|@{}c@{}}
        \smash{\overbrace{\vphantom{\vec{s}_{M_{\mathcal L}}}\hphantom{\lambda_{M_{\mathcal L}}}}
            ^{\vphantom{\vec{\mu}\underline{s}}\smash{\vec{\lambda}}}}
        &
        \smash{\overbrace{\vphantom{\vec{s}_{M_{\mathcal L}}}\hphantom{\vec{s}_{M_{\mathcal L}}}}
            ^{\vphantom{\vec{\mu}\underline{s}}\smash{\underline{s}}}}
        &
        \smash{\overbrace{\vphantom{\vec{s}_{M_{\mathcal L}}}\hphantom{\mu_{M_{\mathcal L}}}}
            ^{\vphantom{\vec{\mu}\underline{s}}\smash{\vec{\mu}}}}
        \\[-3ex]
        \lambda_1 & \vec{s}_1 & \mu_1 \\
        \vdots & \vdots & \vdots \\
        \lambda_{M_{\mathcal L}} & \vec{s}_{M_{\mathcal L}} & \mu_{M_{\mathcal L}}
    \end{array}
    \right]
    \endgroup,
\end{align}
with $\vec{s}_m \in \mathbb{Z}_d^{4n}$. 
Here $\vec{\lambda}$ collects the complex prefactors, $\underline{s}$ is the matrix whose rows are superoperator labels $\vec{s}_m=(\vec{p}_m^{\rm \, L},\vec{p}_m^{\rm \, R})$, and $\vec{\mu}$ stores discrete phases $\mu_m\in\mathbb{Z}_{2d}$, in direct analogy with the Hamiltonian representation. More explicitly, we can separate contributions in the expanded tableau representation in terms of the coherent (Hamiltonian) and incoherent (dissipative) terms:
\begin{align}
    \underline{\mathcal L}
    &=
    \left[
    \begin{array}{@{}c@{\;}|@{\;}cc@{\;}|@{\;}cc@{\;}|@{\;}c@{}l@{}}
        \vec \lambda
        &
        \underline{x}^{\,L}
        &
        \underline{x}^{\,R}
        &
        \underline{z}^{\,L}
        &
        \underline{z}^{\,R}
        &
        \vec \mu
        \\
        \hline
        -{\rm i}c_1/\hbar
        &
        \vec x_1
        &
        \vec 0
        &
        \vec z_1
        &
        \vec 0
        &
        \eta_1
        \\
        +{\rm i}c_1/\hbar
        &
        \vec 0
        &
        -\vec x_1
        &
        \vec 0
        &
        -\vec z_1
        &
        \eta_1-2\vec z_1\cdot\vec x_1
        \\
        \vdots
        &
        \vdots
        &
        \vdots
        &
        \vdots
        &
        \vdots
        &
        \vdots
        \\
        -{\rm i}c_M/\hbar
        &
        \vec x_M
        &
        \vec 0
        &
        \vec z_M
        &
        \vec 0
        &
        \eta_M
        \\
        +{\rm i}c_M/\hbar
        &
        \vec 0
        &
        -\vec x_M
        &
        \vec 0
        &
        -\vec z_M
        &
        \eta_M-2\vec z_M\cdot\vec x_M
        \\
        \hline
        \gamma_1
        &
        \vec x_1^{\,J}
        &
        \vec x_1^{\,J}
        &
        \vec z_1^{\,J}
        &
        \vec z_1^{\,J}
        &
        0
        \\
        -\gamma_1
        &
        \vec 0
        &
        \vec 0
        &
        \vec 0
        &
        \vec 0
        &
        0
        \\
        \vdots
        &
        \vdots
        &
        \vdots
        &
        \vdots
        &
        \vdots
        &
        \vdots
        \\
        \gamma_{M_J}
        &
        \vec x_{M_J}^{\,J}
        &
        \vec x_{M_J}^{\,J}
        &
        \vec z_{M_J}^{\,J}
        &
        \vec z_{M_J}^{\,J}
        &
        0
        \\
        -\gamma_{M_J}
        &
        \vec 0
        &
        \vec 0
        &
        \vec 0
        &
        \vec 0
        &
        0
    \end{array}
    \right]
    \begin{array}{@{}l@{}}
        \\[-0ex]
        \left.
        \begin{array}{@{}c@{}}
            \vphantom{-{\rm i}c_1/\hbar} \\
            \vphantom{+{\rm i}c_1/\hbar} \\
            \vphantom{\vdots} \\
            \vphantom{-{\rm i}c_M/\hbar} \\
            \vphantom{+{\rm i}c_M/\hbar}
        \end{array}
        \right\}
        \underline{\mathcal L}_{\rm coh}
        \\[1.0ex]
        \left.
        \begin{array}{@{}c@{}}
            \vphantom{\gamma_1} \\
            \vphantom{-\gamma_1} \\
            \vphantom{\vdots} \\
            \vphantom{\gamma_{M_J}} \\
            \vphantom{-\gamma_{M_J}}
        \end{array}
        \right\}
        \underline{\mathcal L}_{\rm inc}
    \end{array}.
\end{align}
Note that we have kept the $\pm{\rm i}$ factors in the coherent coefficients and the $\pm$ term in the incoherent part for clarity; equivalently, these factors may be absorbed into the complex prefactors $\lambda_m$.
In this representation, each coherent Hamiltonian term $c_i e^{i\pi \eta_i/d}\hat{P}(\vec{p}_i)$ contributes two Pauli superoperator terms through the commutator,
\begin{equation}
\begin{aligned}
    -\frac{{\rm i}}{\hbar}
    \left[
        c_i e^{{\rm i}\pi\eta_i/d}\hat{P}(\vec{p}_i),
        \hat{\rho}
    \right]
    &=
    -\frac{{\rm i}c_i}{\hbar}
    e^{{\rm i}\pi\eta_i/d}
    \hat{\mathcal P}(\vec{p}_i,\vec{0})|\hat{\rho}\rangle \!\rangle
    \\
    &\quad
    +
    \frac{{\rm i}c_i}{\hbar}
    e^{{\rm i}\pi(\eta_i-2\vec z_i\cdot\vec x_i)/d}
    \hat{\mathcal P}(\vec{0},-\vec{p}_i)|\hat{\rho}\rangle \!\rangle .
\end{aligned}
\end{equation}
The additional phase in the second term follows from the canonical ordering convention $\hat{P}([\vec{x}\,|\,\vec{z}])=\hat{X}^{\vec{x}}\hat{Z}^{\vec{z}}$, for which $\hat{P}(-\vec p_i)^\dagger=e^{2{\rm i}\pi\vec z_i\cdot\vec x_i/d}\hat{P}(\vec p_i)$. Each incoherent Pauli jump contributes
\begin{equation}
    \gamma_\beta
    \left(
        \hat{P}(\vec{q}_\beta)\hat{\rho}\hat{P}(\vec{q}_\beta)^\dagger
        -
        \hat{\rho}
    \right)
    =
    \gamma_\beta\,\hat{\mathcal P}(\vec{q}_\beta,\vec{q}_\beta)|\hat{\rho}\rangle \!\rangle
    -
    \gamma_\beta\,\hat{\mathcal P}(\vec{0},\vec{0})|\hat{\rho}\rangle \!\rangle,
\end{equation}
with any discrete phases absorbed into $\vec{\mu}$ and arbitrary complex prefactors absorbed into $\vec{\lambda}$ as above.

Thus, in addition to the Hamiltonian $\underline{H}$ in Eq.~\eqref{eq:pauli_sum}, a set of GKSL superoperators in Eq.~\eqref{eq:pauli_lindblad_sum} fully characterizes the dissipative model considered. To find a symmetry of this model, we generalize the tableau to the Liouville-space representation on $2n$ qudits, with Paulis acting both on the left and on the right, and combine the coherent and incoherent superoperators in a single tableau with phases tracked together. The same permutation/symplectic symmetry machinery can then be
applied directly to the superoperator rows $\vec{s}$, with symplectic
products computed using the Liouville-space pairing defined above.
A symmetry in this doubled space should be understood as a different object to that for a closed system. Indeed, it is not in itself necessarily a symmetry of the Hamiltonian. However, it identifies invariant operator-space sectors of the open-system dynamics. 
In the following, we retain the notation of the Hamiltonian formulation for simplicity; due to the above construction, the approach follows identical lines for Pauli-GKSL open systems by  $(\vec{\lambda},\underline{s},\vec{\mu})$, exactly analogously to the Hamiltonian $(\vec{c},\underline{p},\vec{\eta})$, with the symplectic product replaced by the Liouville-space form $\langle\cdot,\cdot\rangle_{\mathcal L}$.

% Finally, we can see that the above approach may be easily generalised beyond the Pauli-GKSL case to arbitrary GKSL form dissipators (or even beyond the Markovian cases to arbitrary time-independent superoperators in terms of $x$ and $z$ components). For this it can be seen that replacing the Pauli string with a Pauli sum in Eq.~\eqref{eq:jump_op_Pauli} yields an arbitrary jump operator. Continuing exactly as above with this replacement then yields an Eq.~\eqref{eq:pauli_dissipator} where (1) different left- and right-hand operators $\hat L_{\beta_1}, \hat L_{\beta_2}$ may act in the first term, and (2) the second term may have a left hand operator acting on $\hat \rho$, rather than the density operator alone. Further, each term may have a phase $\zeta_{\beta_1} - \zeta_{\beta_2}$. This nonetheless yields a Liouvillian tableau of the form Eq.~\eqref{eq:pauli_lindblad_sum}.

Finally, the restriction to Pauli-GKSL generators is not essential. 
The simplification in Eq.~\eqref{eq:pauli_dissipator} used the fact that
a Pauli jump operator is unitary, so that $\hat L_\beta^\dagger \hat L_\beta$
is proportional to the identity. For a general GKSL jump operator, one may
instead expand
\begin{equation}
    \hat L_\beta
    =
    \sum_a \ell_{\beta a}
    e^{{\rm i}\pi \xi_{\beta a}/d}
    \hat P(\vec q_{\beta a})
\end{equation}
in the Pauli basis. Substituting this expansion into the dissipator produces
a sum of terms of the form
\begin{equation*}
    \hat P(\vec q_{\beta a})\hat\rho\hat P(\vec q_{\beta b})^\dagger,
    \quad
    \hat P(\vec q_{\beta a})^\dagger \hat P(\vec q_{\beta b})\hat\rho,
    \quad
    \hat\rho\,\hat P(\vec q_{\beta a})^\dagger \hat P(\vec q_{\beta b}),
\end{equation*}
with the corresponding coefficients and discrete phases absorbed into the
tableau. Since products of Pauli strings are again Pauli strings up to a
phase, each contribution is a Pauli superoperator of the form introduced in
Eq.~\eqref{eq:pauli_lindblad_sum}. Thus arbitrary GKSL generators, and more
generally any time-independent superoperator expanded in the Pauli
superoperator basis, can be represented by a Liouvillian tableau of the same
form.

\section{Building the graph}\label{app:building_graph}

\subsection{Simple colours: Coefficients and commutation structure}
In Ref.~\cite{main} we introduced the form of the graph representation used to find valid permutations on the Hamiltonian $\underline{H}$ in Eq.~\eqref{eq:pauli_sum}. Here, we take a more intuitive approach and build the graph for an example case, highlighting the key role of symplectic invariants and the linear dependence structure, showing that they can be efficiently represented on said graph.
We have seen that a symmetry can in general be encoded via a permutation of the rows of a Hamiltonian tableau [see Eq.~\eqref{eq:symmetry_definition_tableau}], what remains is to find an efficient approach to obtain a permutation which can be implemented through a Clifford. We must therefore encode all Clifford invariants into our search structure for $\underline{\Pi}$. 
Two such Clifford invariants, as mentioned above, are the coefficients $c_i$ [Eq.~\eqref{eq:pauli_sum}] (up to the phases in $\eta_i$), and Pauli commutations $\langle \vec{p}_i, \, \vec{p}_j\rangle$. For clarity, in the next sections we will not include the linear dependence structure or phases, which will be treated in Sec.~\ref{app:graph_augmentation} (see also Appendix~\ref{app:linear_codes}) and \ref{app:phase_correction}, respectively. 

Let us consider the example of four Paulis $\vec{p}_i$ with $i \in [1, 2, 3, 4]$. A possible permutation of this set is:
\begin{equation*}
\begin{tikzcd}[row sep=0.25em, column sep=2.6em, cells={inner sep=1pt}]
p_1 \; \arrow[r, mapsto] \; &  \; p_2 \\
p_2  \; \arrow[r, mapsto]  \; &  \; p_1 \\
p_3  \; \arrow[r, mapsto] \;  & \;  p_4 \\
p_4 \;  \arrow[r, mapsto]  \; & \;  p_3
\end{tikzcd}
\end{equation*}
One can check whether this permutation constitutes a symmetry by building the corresponding permutation matrix $\underline{\Pi}$ and symplectic $\underline{F}$ from Eq.~\eqref{eq:basis_map}, and checking Eq.~\eqref{eq:symmetry_definition_tableau} (recalling that $\vec{c}$ is unchanged by Clifford operations). However, performing this check for every permutation is prohibitively expensive for many Paulis. We thus add the following two features to our representation.

First, noting that the coefficients $\vec{c}$ cannot be altered by the action of a Clifford, we use them to colour each index, and collect these colours into the set $\mathcal{X}_V$. If $\vec{p}_1$ and $\vec{p}_2$ have the same colour/coefficient, $\mathcal{X}_V(1) = \mathcal{X}_V(2)$, the others each differ, and we see that the previous permutation
\begin{equation*}
\begin{tikzcd}[row sep=0.25em, column sep=2.6em, cells={inner sep=1pt}]
\bnum{p_1}\; \arrow[r, mapsto] \; &  \; \bnum{p_2} \\
\bnum{p_2}  \; \arrow[r, mapsto]  \; &  \; \bnum{p_1} \\
\gnum{p_3} \; \arrow[r, mapsto] \;  & \;  \rnum{p_4} \\
\rnum{p_4} \;  \arrow[r, mapsto]  \; & \;  \gnum{p_3}
\end{tikzcd}
\end{equation*}
is ruled out, as $\vec{p}_3$ and $\vec{p}_4$ can only map to themselves, $\mathcal{X}_V(3) \neq \mathcal{X}_V(4)$. We can thus use the colours to efficiently prune the search. Identifying each Pauli $\vec{p}_i$ with a vertex of a graph with colour $\mathcal{X}_V(i) = c_i$, we only try permutations that map a vertex to another with the same colour. 

Second, we add commutation information between each vertex, creating a directed graph with edges labelled by
\begin{equation}\label{eq:sp_on_edges}
    \mathcal{X}_E(i\mapsto j)
    =
    E_{ij}
    =
    \langle \vec{p}_i,\vec{p}_j\rangle 
    \bmod{d}.
\end{equation}
(For $d=2$ this may be viewed as an undirected edge colour since $E_{ij}=E_{ji}$.)  The key point is that these edge labels impose pairwise compatibility constraints on a permutation, analogous to how coefficient colours impose single-vertex compatibility. Concretely, if a permutation maps vertices $i\mapsto \Pi(i)$ and $j\mapsto \Pi(j)$, then it must also preserve the edge colour between them:
\begin{equation}\label{eq:edge_color_constraint}
    \mathcal{X}_E(i\mapsto j)
    =
    \mathcal{X}_E\bigl(\Pi(i) \mapsto \Pi(j)\bigr)
    \qquad \text{for all } i,j.
\end{equation}
This provides another pruning rule. For example, suppose we have already fixed that $1\mapsto a$ for some index $a$. For any other vertex $j$, the image $\Pi(j)$ cannot be chosen freely from the coefficient-matched colour class: it must additionally reproduce the commutation pattern with the already-mapped vertex,
\begin{equation}
    E_{1j} \;=\; E_{a,\Pi(j)}.
\end{equation}
Thus, once one vertex is mapped, the outgoing (and incoming) edge colours from that vertex act like a ``fingerprint'' that must be matched by the image vertex. In practice, as more vertices are assigned, the constraints in Eq.~\eqref{eq:edge_color_constraint} rapidly shrink the remaining candidate images, since each newly fixed pair $(i,j)$ enforces one more commutation-colour constraint.

Continuing the above example, suppose the commutation data is encoded onto edges via Eq.~\eqref{eq:sp_on_edges}, with
$
E_{12}=0,\;
E_{13}=E_{23}=1,\;
E_{14}=E_{24}=0,\;
E_{34}=1,
$
and $E_{ij}=E_{ji}$ (qubits).  Then the permutation $\underline{\Pi}= 1 \leftrightarrow 2$ is compatible with the edge colours, because for every pair $\{i,j\}$ we have
$
E_{ij}=E_{\Pi(i)\,\Pi(j)}.
$
In particular, swapping $1$ and $2$ preserves the entire commutation graph to the other vertices: $E_{13}=E_{23}$ and $E_{14}=E_{24}$,
\begin{equation*}
\begin{tikzpicture}[
    vtx/.style={circle, draw, inner sep=1.5pt, minimum size=5mm},
    lab/.style={font=\scriptsize},
    baseline={(current bounding box.center)}
]
% ---------- Left: original graph ----------
\begin{scope}
  \node[vtx] (a1) at (0,0)    {\bnum{1}};
  \node[vtx] (a2) at (2,0)    {\bnum{2}};
  \node[vtx] (a3) at (0,-1.6) {\gnum{3}};
  \node[vtx] (a4) at (2,-1.6) {\rnum{4}};

  % all 6 edges (0 = solid, 1 = dashed)
  \draw[thick]        (a1)--(a2) ;% node[midway,above,lab] {$0$}; % E12
  \draw[thick,dashed] (a1)--(a3);%  node[midway,left, lab] {$1$}; % E13
  \draw[thick]        (a1)--(a4);%  node[midway,sloped,above,lab] {$0$}; % E14
  \draw[thick,dashed] (a2)--(a3);%  node[midway,sloped,above,lab] {$1$}; % E23
  \draw[thick]        (a2)--(a4) ;% node[midway,right,lab] {$0$}; % E24
  \draw[thick,dashed] (a3)--(a4);%  node[midway,below,lab] {$1$}; % E34
\end{scope}

% ---------- Middle: mapsto ----------
\node at (3.15,-0.8) {$\mapsto$};

% ---------- Right: permuted graph ----------
\begin{scope}[xshift=4.4cm]
  % same positions, but labels permuted by Pi=(12)
  \node[vtx] (b1) at (0,0)    {\bnum{2}}; % image of 1
  \node[vtx] (b2) at (2,0)    {\bnum{1}}; % image of 2
  \node[vtx] (b3) at (0,-1.6) {\gnum{3}};
  \node[vtx] (b4) at (2,-1.6) {\rnum{4}};

  % same complete edge-coloured pattern (must match under Pi)
  \draw[thick]        (b1)--(b2);%  node[midway,above,lab] {$0$};
  \draw[thick,dashed] (b1)--(b3);%  node[midway,left, lab] {$1$};
  \draw[thick]        (b1)--(b4);%  node[midway,sloped,above,lab] {$0$};
  \draw[thick,dashed] (b2)--(b3);%  node[midway,sloped,above,lab] {$1$};
  \draw[thick]        (b2)--(b4);%  node[midway,right,lab] {$0$};
  \draw[thick,dashed] (b3)--(b4) ;% node[midway,below,lab] {$1$};
\end{scope}
\end{tikzpicture}.
\end{equation*}

More formally, we consider a fully connected \footnote{in practice it is often useful to create a sparse graph, in which case we can choose the most common edge data type, and remove it. Then no coupling encodes that data (e.g. commutation). We choose to present the fully connected case simply as a convention.} coloured directed graph $G=(V,E,\mathcal{X}_V,\mathcal{X}_E)$, where $V=[1,\, \dots,\, M]$ labels the Pauli strings and $E=V \times V$ are ordered pairs. The vertex colours $\mathcal{X}_V$ include at least (see Appendix~\ref{app:WL}) the coefficients $\vec{c}$ in Eq.~\eqref{eq:pauli_sum}. The edge labels $\mathcal{X}_E(i \mapsto j)$ are determined by the symplectic product $\langle \vec{p}_i,\vec{p}_j\rangle\pmod d$ defined in Eq.~\eqref{eq:symplectic_product_def}. For $d=2$ the symplectic form is symmetric, so one may equivalently view $E$ as the undirected complete graph with colour $\mathcal{X}_E(i \mapsto j) = \mathcal{X}_E(j \mapsto i) = \langle \vec{p}_i, \, \vec{p}_j \rangle$.

In this setting, a candidate permutation $\underline{\Pi}$ is required to preserve the vertex labels (coefficients) and the edge colours (commutation data). A valid one is given in the previous figure. However, this is not yet sufficient to guarantee a Hamiltonian symmetry, as it must also preserve (a) the linear dependence structure of the Pauli strings and (b) phase compatibility. In the next Sec.~\ref{app:phase_correction} and Appendix~\ref{app:linear_codes}, we will investigate how to enforce these through a code automorphism condition as well as additional constraints.

\subsection{Pauli dependency structure via graph augmentation}
\label{app:graph_augmentation}

As noted in Sec.~\ref{app:clifford_summary}, the linear action of Clifford gates on Pauli label vectors implies that the linear dependence structure of the set $\{\vec{p}_i\}_{i=1}^M$ in Eq.~\eqref{eq:pauli_sum} is preserved by any symmetry. Consequently, a candidate permutation $\underline{\Pi}$ must preserve not only the coefficients and commutation data encoded in the graph of Sec.~\ref{app:building_graph}, but also the dependence relations between the Pauli labels. A direct way to incorporate this information is to augment the graph with additional vertices representing minimal dependent subsets.

To illustrate the idea, we use the same four-Pauli example as above and assume there is a single dependence relation over $\mathbb{Z}_2$,
\begin{equation}\label{eq:toy_dep}
    \vec{p}_4=\vec{p}_1+\vec{p}_2
    \qquad\Longleftrightarrow\qquad
    \vec{p}_1+\vec{p}_2+\vec{p}_4=0 .
\end{equation}
The support $C=\{1,2,4\}$ is then a \emph{circuit}, i.e., a minimal dependent set. We encode this by adjoining a new vertex $v_C$ and connecting it to the Pauli vertices in $C$. This new vertex and its edges will be distinguished by a dedicated circuit-type colour and edge labels, respectively. 
%The Pauli vertices retain the coefficient colours and commutation edges introduced in Sec.~\ref{app:building_graph}. 
In this way, any admissible permutation must preserve the incidence pattern between $v_C$ and the Pauli vertices: it may map $\{1,2,4\}$ only to another circuit in the instance.
Continuing the previous example, the augmented graph becomes:

\begin{equation*}
\begin{tikzpicture}[
    vtx/.style={circle, draw, inner sep=1.2pt, minimum size=5mm},
    cvtx/.style={circle, draw, inner sep=1.2pt, minimum size=5mm, fill=green!10},
    mem/.style={thick, green!50!black},
    baseline={(current bounding box.center)}
]

% ---------- Left: original augmented graph ----------
\begin{scope}[scale=0.82]
  % vertices
  \node[vtx]  (a1) at (0,0)       {\bnum{1}};
  \node[vtx]  (a2) at (1.7,0)     {\bnum{2}};
  \node[vtx]  (a3) at (0,-1.35)   {\gnum{3}};
  \node[vtx]  (a4) at (1.7,-1.35) {\rnum{4}};
  \node[cvtx] (ac) at (0.85,1.15) {$v_C$};

  % commutation edges: 0 = solid, 1 = dashed
  \draw[thick]        (a1)--(a2);   % E12=0
  \draw[thick,dashed] (a1)--(a3);   % E13=1
  \draw[thick]        (a1)--(a4);   % E14=0
  \draw[thick,dashed] (a2)--(a3);   % E23=1
  \draw[thick]        (a2)--(a4);   % E24=0
  \draw[thick,dashed] (a3)--(a4);   % E34=1

  % circuit-membership edges
  \draw[mem] (ac)--(a1);
  \draw[mem] (ac)--(a2);
  \draw[mem] (ac)--(a4);

  % label
  \node[font=\scriptsize, green!40!black] at (0.85,1.7) {$C=\{1,2,4\}$};
\end{scope}

% ---------- Middle: mapsto ----------
\node at (3.2,-0.05) {$\mapsto$};

% ---------- Right: permuted augmented graph ----------
\begin{scope}[xshift=4.2cm, scale=0.82]
  % example relabelling (14), ignoring coefficients for illustration of dependence preservation
  \node[vtx]  (b1) at (0,0)       {\rnum{4}};
  \node[vtx]  (b2) at (1.7,0)     {\bnum{2}};
  \node[vtx]  (b3) at (0,-1.35)   {\gnum{3}};
  \node[vtx]  (b4) at (1.7,-1.35) {\bnum{1}};
  \node[cvtx] (bc) at (0.85,1.15) {$v_C$};

  % commutation pattern preserved under relabelling
  \draw[thick]        (b1)--(b2);
  \draw[thick,dashed] (b1)--(b3);
  \draw[thick]        (b1)--(b4);
  \draw[thick,dashed] (b2)--(b3);
  \draw[thick]        (b2)--(b4);
  \draw[thick,dashed] (b3)--(b4);

  % circuit-membership edges
  \draw[mem] (bc)--(b1);
  \draw[mem] (bc)--(b2);
  \draw[mem] (bc)--(b4);

  \node[font=\scriptsize, green!40!black] at (0.85,1.7) {$\Pi(C)=\{4,2,1\}=C$};
\end{scope}

\end{tikzpicture},
\end{equation*}
confirming that the permutation $\underline{\Pi}= 1 \leftrightarrow 2$ preserves the Pauli dependencies.

This example also illustrates that the dependence structure is basis-independent. For instance, temporarily ignoring coefficient constraints, the relabelling $1\leftrightarrow 4$ preserves the circuit $\{1,2,4\}$ even though it exchanges a basis element with a dependent one. What changes is only the choice of basis used to describe the same dependence structure. 

As the systems scale up in size, the presence of \emph{minimal} in the definition of a circuit becomes essential. In this context, minimal means that it is not possible to remove vertices from a circuit while satisfying $\sum_i \vec{p}_i = 0 \bmod{d}$ for the remaining elements $\vec{p}_i$ within that circuit. If we were to encode non-minimal dependent sets as additional constraint vertices, we would record redundancies while rapidly increasing the size of the augmented graph.

However, to ensure that (1) Pauli dependencies are respected and (2) we do not miss some allowed permutation $\underline{\Pi}$ (e.g., one circuit maps to another one, but we do not include either in our graph), we \emph{must} find and use \emph{all} minimal circuits.
In practice, the circuit-augmented graph is constructed from the linear dependencies of the Paulis. Given the tableau $\underline{p}$, we first form a matroid representation $\underline{G} = [\underline{\mathbb1}_{n_b} | \underline{D}]$, with $\underline{D}$ a matrix of dependencies in the chosen basis. The columns of $\underline{G}$ represent the Pauli terms in a chosen independent coordinate basis (see Appendix~\ref{app:linear_codes}). The dependencies among Pauli terms are then obtained from the nullspace of $\underline{G}^{T}$: a vector satisfying $\vec{a}\underline{G}^\top = 0$ specifies a subset of Pauli terms whose linear combination vanishes. 

For qubit systems $d=2$, we enumerate all nonzero binary combinations of a basis for this nullspace, convert each dependency vector into its support, and sort these supports by size. We then retain only those supports which do not contain any previously retained smaller dependent support. The resulting inclusion-minimal dependent supports are precisely the minimal circuits. 
In the augmented graph, one auxiliary vertex is added for each such minimal circuit, with incidence edges connecting it to the Pauli-term vertices in its support. 
For $d>2$, the same construction generalises by replacing binary supports with coefficient-labelled dependencies over $\mathbb Z_d$. A circuit is then an inclusion-minimal support of a nonzero dependency relation, with nonzero coefficients $a_i\in\mathbb Z_d$ retained as edge labels in the augmented graph. \footnote{The coefficients of a dependency relation are defined only up to multiplication by a common nonzero element of $\mathbb Z_d$. In the graph representation these labels should therefore be compared up to this common rescaling.} Thus the qubit case records only the support of each dependency, while the qudit case additionally records the finite-field coefficients
of the corresponding relation.

In general, the number of circuits can be very large. Indeed, with the nullity \cite{Roman2008AdvancedLinearAlgebra} of $\underline p$ being
\begin{equation*}
    M-\operatorname{rank}(\underline p),
\end{equation*}
in the worst case scenario the number of minimal circuits grows exponentially in this quantity. Therefore, explicitly constructing one auxiliary vertex for every circuit may be computationally infeasible, both because the circuits must first be found and because the resulting augmented graph may be too large to be even stored. A possible mitigation is to store a subset of circuits (e.g., smaller ones), potentially limiting the symmetries that can be found. 

Alternatively, we implement in Appendix~\ref{app:linear_codes} another graph representation that avoids explicit circuit enumeration. This approach uses a graph construction without auxiliary vertices, at the cost of performing an (efficient) check of Pauli dependencies for candidate permutations, rejecting them if necessary. For the numerical results in Sec.~\ref{app:numerical_results}, we use this check on candidate permutations unless otherwise stated, rather than adding auxiliary vertices. We choose to present circuits here due to their clarity.

\subsection{Phase correction}\label{app:phase_correction}

As mentioned in Sec.~\ref{app:clifford_summary}, a (symmetry) Clifford unitary $\hat{S}$ is specified by
\begin{equation}
    \hat{S} \cong (\underline{S},\vec{\phi}).
\end{equation}
For $\hat{S}$ to be a symmetry, on top of all constraints discussed in the previous sections, the phase vector $\vec{\phi}$ must ensure that the $\vec{\eta}$ in Eq.~\eqref{eq:pauli_sum} are preserved. 
Although it is in principle possible to encode phase information in the graph directly, as phases are not Clifford invariant, any search scheme for a permutation $\Pi$ that tracks phase information requires expensive computational checks that would considerably slow down our GA finding algorithm.
As such, after a candidate $\underline{S}$ is found, we check whether the phase update resulting from Eq.~\eqref{eq:phase_update} is compatible with the permutation $\underline{\Pi}$ associated to $\underline{S}$, see Eq.~\eqref{eq:PiPF_intro}. 

%Eqs.~\eqref{eq:symplectic_action} and Eq.~\eqref{eq:phase_update} leaves $\hat H$ invariant. For a fixed symplectic $\underline{S}$, the permutation of the exponent matrix $\underline{p}$ is determined from Eq.~\eqref{eq:PiPF_intro}, so the remaining task is to match tableau phases.

To do so, we first choose a reference phase vector $\vec{\phi}_0$ compatible with $\underline{S}$ by solving the parity constraint \cite{Hostens_stabilizer_2005}
\begin{equation}\label{eq:phase_vector_validity}
    (d - 1)\,\mathrm{diag}(\underline{S}\, \underline{\Omega}_l \,\underline{S}^\top) - \vec{\phi}_{0}
    = 0 \bmod{2}.
\end{equation}
This equation does not determine $\vec{\phi}_{0}$ uniquely, but yields a valid $\vec{\phi}_0$ for the given $\underline{S}$. The role of $\vec{\phi}_0$ is only to provide a starting point that is consistent with $\underline{S}$; the remaining phase freedom is then used to correct any residual mismatch. 
Indeed, let us call $\vec{\eta}_0$ the phases of the resulting Hamiltonian after $(\underline{S},\vec{\phi}_0)$ is applied [Eqs.~\eqref{eq:symplectic_action} and \eqref{eq:phase_update}]. If, taking into account the permutation $\underline{\Pi}$ in Eq.~\eqref{eq:PiPF_intro}, $\vec{\eta}_0 = \vec{\eta}$, then $(\underline{S},\vec{\phi}_0)$ is a valid symmetry. If not, we determine $\vec{\phi}$ solving (when possible) the following equations
\begin{equation}\label{eq:affine_phase_row_aligned}
    \vec{\eta}
    =
    \vec{\eta}_{0}
    +
    \underline{p}\,(\vec{\phi}-\vec{\phi}_0)
    \bmod{2d}.
\end{equation}
To do so, we use Eq.~\eqref{eq:phase_vector_validity} to write $\vec{\phi} - \vec{\phi}_0 = 2\vec{u}$ with $\vec{u}\in\mathbb{Z}_d^{2n}$ (which is equivalent to saying that the phase can only be corrected by the action of a Pauli gate, see Appendix~\ref{app:lemma}). The vector $\vec{u}$ is then found from Eq.~\eqref{eq:affine_phase_row_aligned} as $2\underline{p}\vec{u}= \vec{\eta} - \vec{\eta}_{0} \bmod{2d}$. 

Thus, phase correction reduces to an evenness check on $\vec{\eta} - \vec{\eta}_{0}$ modulo $2d$, followed by a single linear solve over $\mathbb{Z}_d$. From $\vec{u}$ we then reconstruct $\vec{\phi}$ to obtain the Clifford symmetry $\hat{S}\cong(\underline{S},\vec{\phi})$.

\subsection{Finding a graph automorphism}\label{app:algorithm}

We have seen that the problem of finding Clifford symmetries can be mapped to a GA plus a phase correction step. In the following, we will discuss two approaches to obtain candidate automorphisms of the above defined graph. The first, used throughout most of the main text, exploits the \texttt{igraph} \cite{csardi2006igraph} package, which itself uses the \texttt{bliss} \cite{junttila2007engineering} algorithm to obtain a set of generators of the automorphism group of a graph, i.e., the primitive GAs from which any GA of the graph can be built. This method requires a vertex-coloured graph, while edges must be label-free. Thus, we modify the graph described above (which also has edge colours for either $d >2$ or including auxiliary vertices) by inserting into each edge a new coloured (by the edge's label) vertex, and removing all edges' colours. From here, it is not guaranteed that the generators themselves have phase corrections (or satisfy the code automorphism check if using the method of Appendix~\ref{app:linear_codes}), and thus combinations of the generators may also be needed for possible phase corrections.

The second approach, described in detail in Appendix~\ref{app:search_scheme}, is a search through possible permutations of the graph via a Minimum Remaining Values (MRV) heuristic \cite{BacchusVanRun1995}, for which the code is available in Ref.~\cite{sympleq}. This is essentially a search over possible vertex permutations with a heuristic based on assigning a tentative value to whichever vertex has the fewest possible options remaining. If the search finds that no options remain or phase correction is not possible after a candidate permutation is found, it backtracks. If using the method of Appendix~\ref{app:linear_codes}, it also backtracks if the permutation code automorphism check fails.

Both approaches greatly benefit from additional vertex colouring, as it reduces the potential number of permutation candidates. We thus begin by performing a one-dimensional Weisfeiler-Leman (WL) refinement \cite[Appendix~\ref{app:WL}]{WL_refinement}, which uses edge data to increase the number of vertex colours. We note an additional approximate method for further increasing the number of vertex colours in Appendix~\ref{app:heuristic}, which provides faster convergence to the desired symmetry in several scenarios.

We note that in each case we can find \emph{all} Clifford symmetries of a Hamiltonian. For this, the \texttt{bliss} implementation is far more efficient, as it can return all generators of the automorphism group. In cases where each of these are themselves a symmetry one obtains all symmetries efficiently, however if this is not the case, combinations of generators must be checked. In the MRV search all symmetries are found by simply completing the full search scheme. In the following we show results with an early stopping criteria, and for a single symmetry of each model.

\section{Numerical results}\label{app:numerical_results}

\begin{figure*}
    \centering
    \includegraphics[width=0.99\linewidth]{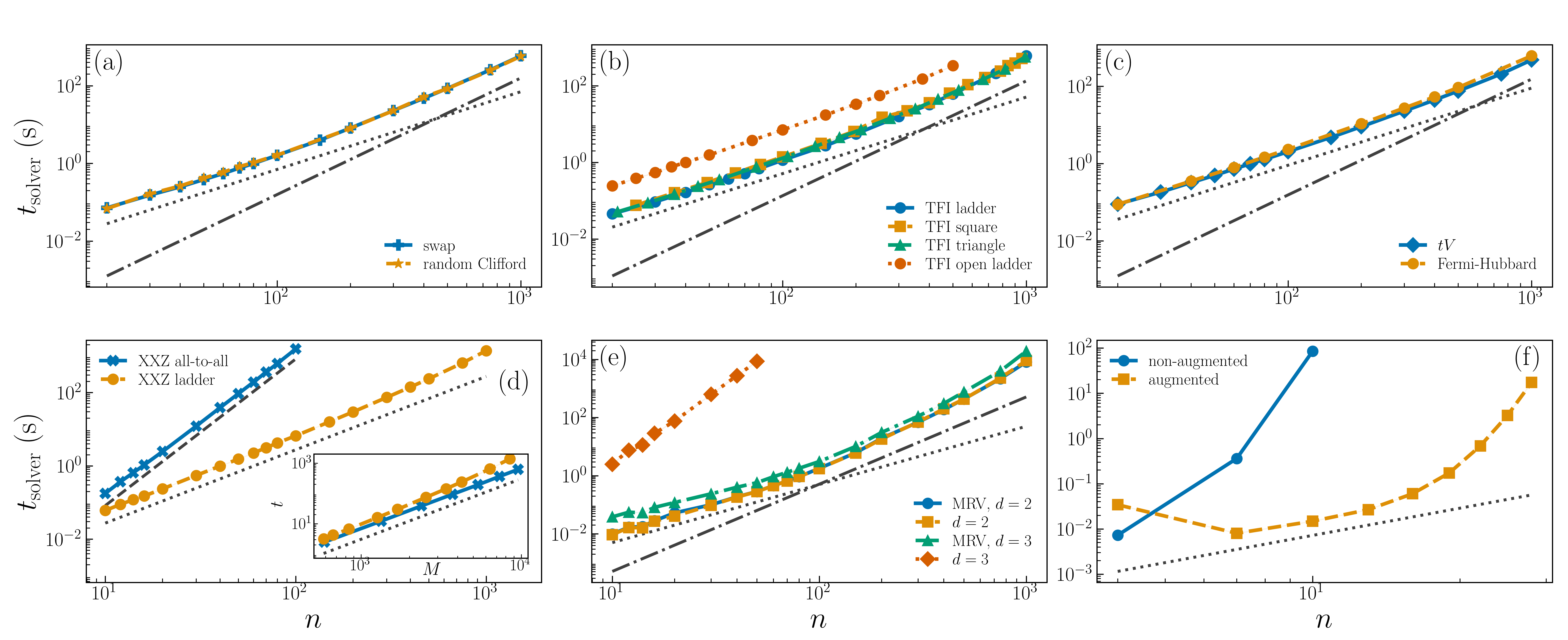}
    \caption{
    Time taken to find Clifford symmetries on the models in Appendix~\ref{app:hamiltonians}.
        (a) Random Pauli Hamiltonian with an injected SWAP (blue) and RCC (yellow, see main text) symmetry.
        (b) TFI model on differing geometries for open and closed systems.
        (c) Fermionic models after Jordan-Wigner transformation (see Appendix~\ref{app:hamiltonians}): the disordered $tV$ ladder (blue), and a Fermi-Hubbard chain (yellow).
        (d) Heisenberg XXZ with local fields on each site for a ladder (yellow) and all-to-all connected (blue) geometries. Scaling with number of Paulis $M$ shown in inset. 
        (e) Highly symmetric Toric code ($c_g = 0)$ for $d=2, 3$, comparing \texttt{bliss} method to MRV search in Appendix~\ref{app:search_scheme}.
        (f) Toric code ($c_g \neq c_z \neq c_x$), comparing augmented graph to non-augmented (with code automorphism check).
        In all panels, grey dotted (dash-dotted) [dashed] lines indicate quadratic (cubic) [quartic] scaling.
        }
    \label{fig:times_all_models}
\end{figure*}

In this section, we use our GA approach to find Clifford symmetries on multiple Hamiltonians. In Sec.~\ref{app:symmetry_finding} we will demonstrate how the time taken to find a symmetry in our approach scales with system size, and in Sec.~\ref{app:symmetry_extension} we show an example of how results of our numerical scheme may be interpreted, to enable analytical representation of the symmetry.
All Hamiltonians used are explicitly written in Appendix~\ref{app:hamiltonians}.

\subsection{Finding Clifford symmetries}\label{app:symmetry_finding}

In this section, we run benchmarking tests on several Hamiltonians. Each model is described in detail in Appendix~\ref{app:hamiltonians}. We use the \texttt{bliss} \cite{csardi2006igraph, junttila2007engineering} implementation, finding the full automorphism group of the (non-augmented) graph and then lift the automorphisms to symmetries via the code automorphism check (see Appendix~\ref{app:linear_codes}) and phase correction step. %In Appendix~\ref{app:search_scheme} we show an explicit search algorithm, and an example application to the Toric code Hamiltonian on qubits and qutrits.

The first model we analyse is a random Pauli Hamiltonian with an injected Clifford symmetry $\hat{S}$. We inject the symmetry by  taking a random Pauli, including in $\hat{H}$ the orbit \footnote{Choosing a random Pauli $\vec{p}$, the orbit is the set $\{pF^i\}_{i \in [0, n_F]}$ such that $pF^{n_F} = p$. We then have that the total orbit, acted upon by any power of $F$ is invariant.} of the Pauli upon action of the symmetry, and repeating the process until the desired number of terms $M$ is (approximately) reached. Afterwards, we `scramble' this Hamiltonian with a large random Clifford circuit (RCC) to ensure that the symmetry is well hidden. Fig.~\ref{fig:times_all_models}(a) shows the time to find a symmetry for an injected SWAP symmetry (blue), and an injected random Clifford symmetry consisting of $10$-gate RCC (yellow) \footnote{
Based on numerical investigation, the number of gates chosen for the RCC symmetry does not have a significant impact on the performance of our algorithm. 
}. We initially see a quadratic scaling (grey dotted line) of the time taken with system size, consistent with the graph building time, which is limited by the $M^2$ symplectic products (here we choose $M \approx 3n$). We observe a slight increase from the $n^2$ scaling for large $n$, which is due to finding $\underline{S}$ from a $\underline{\Pi}$: a subroutine scaling as $n^3$ (grey dashed-dotted line).

We now turn to more realistic models. In Fig.~\ref{fig:times_all_models}(b) we show the times to find symmetries on Transverse Field Ising (TFI) models of differing geometries: ladder (blue), square (yellow), triangular (green) and a ladder with dephasing GKSL terms (orange, see Sec.~\ref{sec:pauli_lindblad}). On the other hand, Fig.~\ref{fig:times_all_models}(c) includes Fermionic examples. A Fermi-Hubbard chain \cite{essler2005thehubbard, McClean_2020} is in yellow, and a disordered Fermionic t-V chain \cite{Vardhan_fermionic_2017, KieferEmmanouilidis_Unlimited_2021} in blue.
In both panels (b) and (c) we observe largely similar scaling to the previous models from (a), with the main notable difference being the open ladder, which is slower. This is explained by the requirement of doubling $n$ to represent the Liouvillian, see Sec.~\ref{sec:pauli_lindblad}.

In Fig.~\ref{fig:times_all_models}(d) we benchmark Heisenberg XXZ models \cite{Ramos_Confinement_2020, Voorden_scars_2020}, with all-to-all couplings (blue) and a ladder (yellow) geometry. Here, we see a significant difference in the scaling of the time taken to find a symmetry, owing to the system properties. The all-to-all model has $M \propto n^2$ Pauli terms, whereas the ladder geometry has $M \propto n$. As the graph size scales with $M$, this yields different behaviours at varying $n$. This is consistent with the $\propto n^4$ (dashed line), and $\propto n^2$ (dotted line) scalings observed. Indeed, in the inset we plot the time taken by our approach against $M$ rather than $n$, and see that both models are characterized by an approximately quadratic scaling.

The final model we analyse is the Toric code \cite{BullockBrennen2007QuditSurface, Resende_toric, Chan_VQE_2024, ferguson2021measurement} with different qudit dimensions $d=2$ and $d=3$, see Eq.~\eqref{eq:toricH} and Sec.~\ref{app:symmetry_extension}. In Fig.~\ref{fig:times_all_models}(e) we consider a highly symmetric parameter regime where there is no local field on the qudits [$c_g = 0$, $c_x \neq c_z$ in Eq.~\eqref{eq:toricH}] and therefore the system can be efficiently diagonalized \cite{ferguson2021measurement}. Here, we compare the MRV search of Appendix~\ref{app:search_scheme} to the \texttt{bliss} method, finding comparable results for qubits ($d=2$), however the MRV search is faster for qutrits ($d=3$). This we ascribe to a combination of factors. Firstly, \texttt{bliss} requires a mapping to a vertex coloured graph only, which is trivial for $d=2$, as we simply have connected or disconnected vertices encoding commutation relations. However for $d=3$ the symplectic product can be 0, 1, or 2, and thus edge labels are required. The resulting effective vertex coloured graph is larger. Second, we see that the MRV search is somewhat fortuitous for the $d=3$ case. Indeed, if we randomise the search, the MRV becomes much slower. This highlights that, ultimately, the time required to find a GA (and thus a Clifford symmetry) is problem dependent~\cite{cameron2003automorphisms}.

We then, in Fig.~\ref{fig:times_all_models}(f), look at the regime with far fewer symmetries [$c_g \neq 0$, $c_x \neq c_z$ in Eq.~\eqref{eq:toricH}]. Here, we do not observe any polynomially scaling search scheme, and see that the \texttt{bliss} scheme is superior. This is the only example where we see graph augmentation with auxiliary vertices for minimal circuits is beneficial, see Sec.~\ref{app:graph_augmentation}. This is because of the complex Pauli dependency relations of the Toric code, and the fact that without augmentation there are a large number of GA generators that do not lift to a symmetry due to the Pauli dependency checks failing. Whilst there are many minimal circuits 
($4008$ for $n=28$, and $10924$ for $n=32$) 
yielding a large graph ($84326$ vertices for $n=28$, and $262748$ for $n=32$), the augmentation prevents such failing checks, which otherwise take a prohibitive time. The limitation of this method is simply the size of the overall graph, which is further increased in the \texttt{bliss} approach by the need to map the input graph described above to a vertex coloured graph only.

It is challenging to pinpoint the exact scaling of our algorithm. There are several subroutines -- computing the graph ($\propto M^2$), determining $\underline{S}$ from $\underline{\Pi}$ ($\propto n^3$), finding a GA (problem dependent, often polynomial, worst case $\propto \exp(\sqrt{(M \log M)})$ \cite{Babai2016GIQuasipoly, McKay2014practical}) -- that are the limiting factor in different regimes. 
It ultimately depends on the input Hamiltonians and its parameters $n$ and $M$, but also coefficients $\vec{c}$, phases $\vec{\eta}$ and interactions $\underline{p}$. In most of the examples above, for $n \leq 1000$, we see an approximate quadratic scaling with $M$ dictated by the graph computation. Concerning the Toric code results in Fig.~\ref{fig:times_all_models}(f), permutations in the non-augmented case are obtained that do not pass the code automorphism checks, and thus we must combine them to ensure correctness of the phases and the Pauli dependencies (see Sec.~\ref{app:graph_augmentation} and \ref{app:phase_correction}). The number of these combinations increases combinatorially, making instances with many qubits $n$ inaccessible. The graph augmentation strategy [yellow squares Fig.~\ref{fig:times_all_models}(f)] with the auxiliary vertices described in Sec.~\ref{app:graph_augmentation} ensures all permutations automatically satisfy Pauli dependencies, hence the significant speedup. The cost, however, is that the memory needed to build the augmented graph is significantly larger, posing another limitation to the attainable qudit numbers. 

\subsection{Inspective extension of symmetries}\label{app:symmetry_extension}

As explained in Ref.~\cite{main}, for a given model it is often possible to use the outcomes of our algorithm to inductively build a symmetry that is valid for arbitrary system sizes. The TFI model on a ladder was considered in Ref.~\cite{main}.
Here, we find a general symmetry of the Toric code in Fig.~\ref{fig:times_all_models}(f), for which our algorithm cannot handle instances with more than $n = 28$ qubits.

The Toric code Hamiltonian is
\begin{equation}\label{eq:toricH}
    \hat H_{\mathrm{TC}}
    =
    c_z \sum_s \hat A_s
    +
    c_x \sum_p \hat B_p
    +
    c_g \sum_e \hat Z_e,
\end{equation}
where $\hat A_s=\prod_{e\in s}\hat Z_e$ (star `$s$' terms) and $\hat B_p=\prod_{e\in\partial p}\hat X_e$ (plaquette `$\partial p$' terms) \cite{BullockBrennen2007QuditSurface, Resende_toric, Chan_VQE_2024, ferguson2021measurement}.
The diagonalizability (star and plaquette terms are all-to-all commuting) is broken by the $Z$-field term proportional to $c_g$ in Eq.~\eqref{eq:toricH}.
In this regime our GA search is slowest [Fig.~\ref{fig:times_all_models}(f)], yet it remains possible to extract exact symmetries at large sizes $n$ by extrapolating the pattern identified for small-$n$ solutions.

We consider an open-boundary $N_x \times N_y = 2 \times L$ ladder (qubits on edges) with $n = N_x(N_y-1)+N_y(N_x-1) = 3L-2$ qubits. The obtained Clifford symmetries have symplectic
\begin{equation}\label{eq:ladder_lengthflip_symp}
\underline{S}_{\rm TC}
=
\begin{pmatrix}
\underline{\Pi}_{\rm TC} & \underline C_{\rm TC}\\[2pt]
\underline 0_n & \underline{\Pi}_{\rm TC}
\end{pmatrix},
\end{equation}
so that, for each $\vec{p}_i = [\vec{x_i} \, | \, \vec{z}_i]$, we have $\vec{x}_i \mapsto \vec{x} \,\underline{\Pi}_{\rm TC}$ and $\vec{z}_i \mapsto \vec{z}_i \,\underline{\Pi}_{\rm TC}+\vec{x}_i\,\underline{C}_{\rm TC}$.
$\underline{S}_{\rm TC}$ is a reflection of the lattice, but acts on vectors $(\vec x,\vec z)\in\mathbb{Z}_2^{2n}$ whose components are indexed by a fixed qubit ordering $e_0,\dots,e_{n-1}$. Thus the reflection is represented by a permutation matrix $\underline\Pi_{\rm TC}$ whose entries depend on this chosen ordering. In the following, we explain how to build the matrices $\underline{\Pi}_{\rm TC}$ and $\underline{C}_{\rm TC}$.

For the open $2\times L$ ladder it is natural to group edges by type,
\begin{equation*}
(h_0,\dots,h_{L-1}\mid v^{\rm L}_0,\dots,v^{\rm L}_{L-2}\mid v^{\rm R}_0,\dots,v^{\rm R}_{L-2}),
\end{equation*}
where $h_j$ are the horizontal rung edges and $v^{\rm L}_j$, $v^{\rm R}_j$ are the vertical rail edges between columns $j$ and $j+1$ on the left and right legs. In this basis the longitudinal reflection acts as
\begin{equation*}
h_j \mapsto h_{L-1-j},\qquad
v^{\rm L/R}_j \mapsto v^{\rm L/R}_{L-2-j},
\end{equation*}
such that $\underline\Pi_{\rm TC}e_i=e_{\pi_{\rm TC}(i)}$ with
\begin{equation}\label{eq:ladder_lengthflip_perm}
\pi_{\rm TC}(i)=
\begin{cases}
L-1-i, & 0 \le i < L\\
3L-2-i, & L \le i < 2L-1\\
5L-4-i, & 2L-1 \le i < 3L-2
\end{cases}.
\end{equation}

On the other hand, the off-diagonal block $\underline{C}_{\rm TC}$ encodes the non-geometric Clifford dressing that mixes $X$-exponents into $Z$-exponents. It is convenient to factor out the permutation by defining
\begin{equation}
\underline K_{\rm TC} = \underline C_{\rm TC}\,\underline\Pi_{\rm TC}^\top,
\qquad\text{so that}\qquad
\underline C_{\rm TC}=\underline K_{\rm TC}\,\underline\Pi_{\rm TC},
\end{equation}
with symplecticity implying $\underline K_{\rm TC}=\underline K_{\rm TC}^\top$ over $\mathbb{Z}_2$.
For the symmetries found here, $\underline K_{\rm TC}$ is sparse and local: it couples each rung to its adjacent rail edges and couples neighbouring rail edges along each leg. Writing $\underline K_{\rm TC}$ in the block ordering $(h\,|\,v^{\rm L}\,|\,v^{\rm R})$ gives
\begin{equation}\label{eq:K_closed_form}
\underline K_{\rm TC} 
=
\begin{pmatrix}
\underline 0_{L} & \underline T & \underline T\\
\underline T^\top & \underline I_{L-1} & \underline I_{L-1}+\underline A\\
\underline T^\top & \underline I_{L-1}+\underline A & \underline I_{L-1}
\end{pmatrix}
\bmod{2},
\end{equation}
where $\underline T\in\mathbb{Z}_2^{L\times(L-1)}$ is the rung-rail incidence matrix with
$T_{j,k}=1$ iff $k=j$ or $k=j-1$, and
$\underline{A}\in\mathbb{Z}_2^{(L-1)\times(L-1)}$ is the path adjacency on $L-1$ sites:
$A_{k,k+1}=A_{k+1,k}=1$. 
%for $k=0,\dots,L-3$ and zero otherwise.
Eqs.~\eqref{eq:ladder_lengthflip_symp}--\eqref{eq:K_closed_form} specify the symplectic for arbitrary $L$.

To fully specify the symmetry $\hat{S}_{\rm TC}$, we only miss the associated phase vector $\vec\phi=(\vec\phi_X,\vec\phi_Z)\in\mathbb Z_4^{2n}$, see Eq.~\eqref{eq:phase_update}. Following the outcomes of our algorithm, 
\begin{equation}\label{eq:ladder_lengthflip_phase}
\vec\phi_Z=\vec 0_n,
\qquad
(\vec\phi_X)_i=
\begin{cases}
1, & 0\le i < L \quad (h\text{-edges})\\
0, & L\le i < n \quad (v\text{-edges})
\end{cases},
\end{equation}
yielding the desired symmetry $\hat{S}_{\rm TC} \cong (\underline{S}_{\rm TC}, \vec{\phi})$ for the Toric code in Eq.~\eqref{eq:toricH} on an arbitrarily long ladder.

\section{Conclusion}\label{app:conclusion}

In this work, we have described and demonstrated a graph automorphism algorithm for finding Clifford symmetries in quantum many-body systems. 
We show that both open and closed quantum systems can be tackled by our approach, with the only requirement being that a representation of the model as a Pauli sum on a set of qudits (potentially of different dimensions) exists.
We have benchmarked our algorithm on a range of models, with different search schemes to investigate their performances. In all but one model we have found symmetries on hundreds of qubits. For the Toric code, the search scheme is less well suited, yet we still found a symmetry on $n = 28$ qubits in less than an hour on a desktop machine. For larger instances $n > 28$, we have analytically obtained symmetries from the results of our algorithm for increasing models' sizes.

Our current search scheme is likely suboptimal for many models, and additional speedups may be helpful in different circumstances. For example, we have seen for the Toric model that in some cases it is useful to use the full augmented graph including all dependency information; an improved approach to building such graphs, without the need to build \emph{all} circuits and then refine, would drastically speed up this process. Additionally, for such systems with complex dependency structure, heuristic methods to build minimal circuits may suffice, offering a potentially incomplete search for symmetries, with a benefit of a much faster search in such systems.

Due to the requirement for phase correction (and optional but often beneficial checks of the Pauli dependence relation to remove need for auxiliary vertices), it remains an open question as to whether finding Clifford symmetries is reducible exactly to GA. We note that phase correction was possible for the majority of our candidate symmetries, and thus, empirically, our examples are no more complex than the underlying GA. However, as phases are not Clifford invariant, any (obvious) search scheme that includes phase information necessarily requires updates to the graph colouring at each step (or regular intervals), distancing the problem from typical GA settings by a single check for possible phase corrections.

\section*{Acknowledgments}
We thank Andrew Jena and Lane Gunderman for discussions that were beyond insightful.
We acknowledge and are grateful to the EPSRC quantum career development grant EP/W028301/1 and the EPSRC Standard Research grant EP/Z534250/1.

\section*{References}
\bibliographystyle{apsrev4-1}
\bibliography{bibli}

\appendix

\appendix

\section{Mixed dimension qudits}\label{app:mixed_d}
The symplectic representation of the Clifford formalism may be extended to mixed dimension qudits by a small modification to the formalism of the main text.
Mixed dimensions require a mixed-modulus tableau, where a tableau row is written as
\begin{equation*}
  p=[x\,|\,z],\qquad x=(x_1,\dots,x_n),\ z=(z_1,\dots,z_n),\quad
  x_i,z_i\in\mathbb Z_{d_i}.
\end{equation*}
Let us define
\begin{equation*}
  L = \mathrm{lcm}(d_1,\dots,d_n),
\end{equation*} 
together with the complex phase roots
\begin{equation*}
  \omega_{d_i} = e^{2\pi i/d_i},\qquad
  \zeta_{d_i} = e^{\pi i/d_i}.
\end{equation*}
With these changes, one can update the tableau representation in Eq.~\eqref{eq:pauli_sum} to account for mixed-qudit systems. 
Concerning the arithmetic operation, one has to update the definition of the symplectic form in Eq.~\eqref{eq:symplectic_form} to
\begin{equation*}
  \underline{W}_{\mathbf d} = \mathrm{diag}\left(\frac{L}{d_1},\dots,\frac{L}{d_n}\right),
  \qquad
  \underline{\Omega}_{\mathbf d} =
  \begin{pmatrix}
    0 & \underline{W}_{\mathbf d}\\[2pt]
    -\underline{W}_{\mathbf d} & 0
  \end{pmatrix}.
\end{equation*}
Note that tableau arithmetic is site-wise in $\mathbb Z_{d_i}$, while phase exponents live in $\mathbb Z_{2L}$.

For finding a Clifford symmetry in practice, it is not necessary to perform significant alterations to the single dimension case, we restrict to the common Clifford operations which do not entangle qudits of different dimensions. Thus, any Clifford symmetry must be localised to a single set of like dimension qudits. For mixed dimensions, then, one can simply search for symmetries independently on each like dimensioned subspace.

\section{Pauli dependency structure via permutation code automorphism}\label{app:linear_codes}

As we have seen in Sec.~\ref{app:graph_augmentation}, the linear dependence structure of the Pauli labels can be encoded directly in the graph by adjoining vertices for the minimal circuits. This provides a conceptually transparent representation, but for some instances it may not be an efficient approach: when the nullity (dimension of the kernel) is sufficiently large, obtaining the full set of circuits by explicit nullspace enumeration can be computationally expensive, and the resulting augmented graph may become large. We therefore seek an alternative formulation that enforces the same dependence constraint without requiring an explicit computation of all circuits.

The example of Eq.~\eqref{eq:toy_dep} already shows the main subtlety. The permutation $1\leftrightarrow 4$ preserves the circuit $\{1,2,4\}$, even though it exchanges a basis element with a dependent one. Thus the invariant object is not the coordinate tuple of a Pauli relative to a chosen basis, but the basis-independent dependence structure itself.

A possible solution is to describe the Pauli labels through a generator matrix for the associated linear code. One may choose a particular basis set $\underline{p}_{b}$ (with $n_b \leq 2n$ elements) and write each dependent Pauli as a combination of elements in that basis. However, this coordinate description is not invariant in itself, since a symmetry is allowed to map one independent generating set to another, thereby changing the representation of a given Pauli in the chosen basis. What is invariant is the row space of the corresponding generator matrix, and this is what leads to the permutation code automorphism condition.

To see this explicitly, we use the same example as Sec.~\ref{app:graph_augmentation}.
One convenient choice of independent set is $\{\vec{p}_1,\vec{p}_2,\vec{p}_3\}$ (so $n_b=3$), in which $\vec{p}_4$ has coordinate tuple $(1,1,0)$ relative to $(\vec{p}_1,\vec{p}_2,\vec{p}_3)$. But the same instance admits a different independent set, for example $\{\vec{p}_4,\vec{p}_2,\vec{p}_3\}$, since $\vec{p}_1=\vec{p}_4+\vec{p}_2$ also follows from Eq.~\eqref{eq:toy_dep}. 
We consider -- ignoring colours temporarily for clarity -- a symmetry encoding the permutation $1\leftrightarrow 4$ (leaving $2$ and $3$ fixed):
\begin{equation*}
\begin{tikzcd}[row sep=0.25em, column sep=2.6em, cells={inner sep=1pt}]
p_1 \; \arrow[r, mapsto] \; &  \; p_4 \\
p_2  \; \arrow[r, mapsto]  \; &  \; p_2 \\
p_3  \; \arrow[r, mapsto] \;  & \;  p_3 \\
p_4 \;  \arrow[r, mapsto]  \; & \;  p_1
\end{tikzcd}
\end{equation*}
Under this relabelling, the relation in Eq.~\eqref{eq:toy_dep} is preserved: the circuit $\{1,2,4\}$ maps to itself, and equivalently the equation $\vec{p}_1+\vec{p}_2+\vec{p}_4=0$ is sent to $\vec{p}_4+\vec{p}_2+\vec{p}_1=0$, i.e., the same constraint. The only difference is the dependencies in a chosen basis. In the new one $\{\vec{p}_4,\vec{p}_2,\vec{p}_3\}$, the previously basis element $\vec{p}_1$ becomes dependent with coordinate tuple $(1,1,0)$.

This illustrates that the dependence structure that must be preserved by a symmetry is basis-independent (it is determined by which subsets are independent), whereas the coordinate tuple of a Pauli relative to an arbitrarily chosen basis is not. In our case, we enforce dependence invariance via a code automorphism test. 
%In practice, we allow for a change of basis performed by an element of the general linear group $\mathrm{GL}(n_b,d)$, i.e., an $n_b \times n_b$ invertible matrix with elements selected from $\mathbb{Z}_d$.
%
%A convenient basis-independent way to \Luca{perform the code automorphism test for linear dependencies} is through linear codes \cite{code_equiv, Sayginel_codeauto_2025}. A linear code is a subspace of $\mathbb{Z}_d^{n_b\times M}$ that is closed under addition and scalar multiplication. 
In practice, we choose $n_b=\mathrm{rank}\{\underline{p}\}$ independent Paulis and build the generator matrix
\begin{equation}\label{eq:G_def_intro}
    \underline{G}\ \in\ \mathbb{Z}_d^{n_b\times M},
\end{equation}
whose $M$ columns are labelled by the Pauli indices $i\in[1,\, \dots,\, M]$. 
This generator matrix can be written as $\underline{G} = [\underline{\mathbb{1}}_{n_b} | \underline{D}]$, with $\underline{D} \in \mathbb{Z}_d^{n_b \times (M - n_b)}$. Here the identity part denotes that the first $n_b$ columns represent Paulis within the basis $\underline{p}_b$, further columns have entries which indicate the dependencies of each subsequent Pauli. In the previous example, 
%
%In the basis of our previous example, the dependence data may be encoded by a generator matrix whose first three columns correspond to the chosen independent set and whose fourth column encodes the relation for $\vec{p}_4$:
\begin{equation}\label{eq:G_example_intro}
    \underline{G}
    \;=\;
    \left[
    \begin{array}{cccc}
        1 & 0 & 0 & 1 \\
        0 & 1 & 0 & 1 \\
        0 & 0 & 1 & 0
    \end{array}
    \right],
\end{equation}
so that column $4$ is the sum of columns $1$ and $2$, implying Eq.~\eqref{eq:toy_dep}. The specific form of Eqs.~\eqref{eq:G_def_intro} and \eqref{eq:G_example_intro} depends on the chosen independent set. However, the row space of $\underline{G}$ (and hence which subsets of columns are independent/dependent) is invariant: any basis change $\underline{U}$
\begin{equation}\label{eq:U_row_intro}
    \underline{G}\ \mapsto\ \underline{U}\,\underline{G},
    \qquad \underline{U}\in \mathrm{GL}(n_b,d),
\end{equation}
$\mathrm{GL}(n_b,d)$ being the general linear group, does not alter the underlying dependence relations. This invariance is what we use to perform the code automorphism test for linear dependencies.

As we have seen in Eq.~\eqref{eq:symmetry_definition_tableau}, a Clifford symmetry corresponds to a permutation $\underline{\Pi}$ of the Hamiltonian, yet not all permutations correspond to symmetries. Concerning the linear dependencies and the matrix $\underline{G}$, $\underline{\Pi}$ also permutes its columns. The dependence structure is preserved when the permuted columns span the same row space, i.e. when there exists a change of basis $\underline{U}$ such that
\begin{equation}\label{eq:code_auto_intro}
    \underline{U}\,\underline{G}\,\underline{\Pi} \;=\; \underline{G}.
\end{equation}
This is the permutation code automorphism condition \cite{Huffman_Pless_2010} implemented by our algorithm. 
%$\underline{\Pi}$ is an automorphism of the linear code defined by $\underline{G}$. 
It is basis-independent because any change of generating set resulting from the permutation $\underline{\Pi}$ can be reverted by the change of basis $\underline{U}$ in Eq.~\eqref{eq:code_auto_intro}.

Going back to our example, we have seen that there are two valid [before the test in Eq.~\eqref{eq:code_auto_intro}] relabellings. The first $\vec{p}_1 \mapsto \vec{p}_2$ corresponds to permutation $\underline{\Pi}_{(12)}$. This preserves the structure Eq.~\eqref{eq:G_example_intro} up to a row permutation  (indeed $\underline{U}$ can be chosen to swap the first two rows), so Eq.~\eqref{eq:code_auto_intro} holds and the dependence structure is compatible with $\underline{\Pi}_{(12)}$.
The second relabelling maps $\vec{p}_1 \mapsto \vec{p}_4$ and corresponds to a permutation $\underline{\Pi}_{(14)}$. This changes the generator matrix $\underline{G}$ more radically, as it moves a dependent column ($4$) into the independent position ($1$). However, with
\begin{equation*}
    \underline{U}
    \;=\;
    \left[
    \begin{array}{ccc}
        1 & 0 & 0 \\
        1 & 1 & 0 \\
        0 & 0 & 1 
    \end{array}
    \right] \in \mathrm{GL}(n_b=3,d=2),
\end{equation*}
it is still possible to satisfy Eq.~\eqref{eq:code_auto_intro}: $\underline{U}\,\underline{G}\,\underline{\Pi}_{(14)}=\underline{G}$. 

By contrast, a relabelling that changes the dependence relations cannot be repaired by any row basis change. For example, if a permutation sends the dependent column (here column $4$) to a column that is not in the span of the images of columns $1$ and $2$, then there exists no $\underline{U}$ satisfying Eq.~\eqref{eq:code_auto_intro}. In this case the permutation is not an automorphism, and cannot correspond to a Hamiltonian symmetry. This example illustrates why Eq.~\eqref{eq:code_auto_intro} is essential: it captures the freedom to choose a different generating set after relabelling.

\section{Phase correction}\label{app:lemma}
\begin{lemma}
\label{lem:prime_d_phase}
Let $d$ be prime, and fix a candidate symplectic map $\underline{S}$. If $\vec{\phi}_0$ is any phase vector compatible with $\underline{S}$, then every phase vector compatible with the same $\underline{S}$ can be written as
\begin{equation}
    \vec{\phi}=\vec{\phi}_0+2\vec{u}
    \bmod{2d},
\end{equation}
for a unique $\vec{u}\in\mathbb{Z}_d^{2n}$.
\end{lemma}

\begin{proof}
For fixed $\underline{S}$, compatibility of the phase vector is imposed by Eq.~\eqref{eq:phase_vector_validity} \cite{Hostens_stabilizer_2005}. 
Let $\vec{\phi}$ and $\vec{\phi}_0$ be two compatible phase vectors for the same $\underline{S}$. Subtracting Eq.~\eqref{eq:phase_vector_validity} for these two vectors gives
\begin{equation}
    \vec{\phi}-\vec{\phi}_0 \equiv 0 \pmod 2.
\end{equation}
Since tableau phases are defined modulo $2d$, each component of $\vec{\phi}-\vec{\phi}_0$ is therefore even residue $2d$.
Each component is thus uniquely of the form $2u$ with $u\in\mathbb{Z}_d$ (equivalently $u\in\{0,\dots,d-1\}$). Hence there exists a unique $\vec{u}\in\mathbb{Z}_d^{2n}$ such that
\begin{equation*}
    \vec{\phi}-\vec{\phi}_0 \equiv 2\vec{u}\pmod{2d},
\end{equation*}
and therefore $\vec{\phi}=\vec{\phi}_0+2\vec{u}\bmod{2d}$.

Conversely, for any $\vec{u}\in\mathbb{Z}_d^{2n}$, the vector $\vec{\phi}_0+2\vec{u}$ has the same parity modulo $2$ as $\vec{\phi}_0$, and therefore satisfies the same compatibility condition for $\underline{S}$.
\end{proof}

Lemma~\ref{lem:prime_d_phase} shows that, for fixed $\underline{S}$, the remaining freedom in the phase vector is exactly that induced by Pauli conjugation.

\section{Adding more colours: Weisfeiler-Leman colour refinement}\label{app:WL}

Colours of vertices narrow down the available options for permutations. The core idea of one-dimensional Weisfeiler-Leman (WL) colour refinement is to iteratively sharpen a vertex colouring by replacing each vertex colour with a canonical encoding of its current colour together with the multiset of colours observed in its neighbourhood \cite{WL_refinement}. In our setting we apply WL to the edge-coloured complete graph associated with the symplectic-product matrix $\underline{M}(\underline{H})$, using it as a fast preprocessing step to obtain a refined vertex partition that accelerates the subsequent search over candidate permutations.

WL refinement is initialised by assigning every vertex $i$ a seed colour $\kappa_i^{(0)}$ that contains at minimum the coefficient label $c_i$. Optionally, additional per-vertex labels may be included (see Appendix~\ref{app:heuristic}). Importantly, this initial colouring is assigned to all vertices simultaneously; a WL `round' updates all colours in parallel.

Given a colouring $\{\kappa_i^{(t)}\}_{i=1}^M$ at round $t$, each vertex $i$ is reassigned a new colour by forming a key consisting of its current colour together with the multiset of edge-coloured neighbour data,
\begin{equation}
\label{eq:wl_update}
    \kappa_i^{(t+1)} \;=\; \Phi\!\left(
        \kappa_i^{(t)},
        \left\{\left(\kappa_j^{(t)},\,\mathcal{X}(i\mapsto j)\right)\right\}_{j=1}^{M}
    \right),
\end{equation}
where $\mathcal{X}(i \mapsto j)$ denotes the edge label (here $\mathcal{X}(i \mapsto j)=\underline{M}(\underline{H})_{ij}$), and $\Phi$ denotes a canonical relabelling map that assigns identical keys the same new colour and different keys different colours. Iteration continues until the colouring stabilises, yielding a stable partition $\mathcal{C}^{(0)}=\{C_\alpha\}$; we denote the final colour of vertex $i$ by $\tilde{c}_i$.

WL refinement is not a complete isomorphism test: non-isomorphic graphs can share the same stable colouring \cite{WL_refinement}. This limitation is not problematic for our purposes, because we do not use WL to decide whether a symmetry exists. Instead, WL is used as a fast way to \emph{increase the effective number of vertex colours} by incorporating local edge-colour information into the vertex labels. Intuitively, two vertices that may have started with the same coefficient colour will often acquire different WL colours once their commutation neighbourhoods are taken into account.

This refined colouring simplifies the symmetry search as a valid relabelling cannot map a vertex to one of a different colour. Equivalently, the stable WL partition splits the vertex set into smaller colour classes, and the backtracking search explores candidate images only within the corresponding class. Any permutation that survives this pruning is subsequently checked against the full exact constraints and discarded if it fails. In this way WL refinement reduces the branching factor of the search without affecting correctness of the final verified symmetries.

\section{Search scheme}\label{app:search_scheme}

\begin{figure}[t]
\centering
\begin{tikzpicture}[
  >=Latex,
  font=\scriptsize,
  node distance=4.2mm and 6mm,
  proc/.style={
    rectangle,
    draw=black!75,
    semithick,
    fill=white,
    align=center,
    inner sep=2.6pt,
    text width=24mm,
    minimum height=6.6mm
  },
  decision/.style={
    diamond,
    draw=black!75,
    semithick,
    fill=white,
    align=center,
    aspect=2.0,
    inner sep=0.8pt,
    minimum width=17mm,
    minimum height=8.5mm
  },
  verify/.style={
    rectangle,
    draw=black!75,
    semithick,
    fill=black!6,
    align=center,
    inner sep=2.6pt,
    text width=24mm,
    minimum height=6.6mm
  },
  terminal/.style={
    rectangle,
    rounded corners=2pt,
    draw=black!75,
    semithick,
    fill=white,
    align=center,
    inner sep=2.6pt,
    text width=21mm,
    minimum height=6.6mm
  },
  arrow/.style={->, semithick},
  lab/.style={font=\footnotesize, fill=white, inner sep=0.8pt}
]

% ----- main column -----
\node[proc]     (init)      {Initialise $\vec{\pi} = (*, *, \dots)$};
\node[proc,     below=of init]      (select)    {Choose next $i$ (MRV)};
\node[proc,     below=of select]    (build)     {Build $C_i$ in Eq.~\eqref{eq:colour_classes}};
\node[decision, below=5.0mm of build]  (candrem)   {Candidates\\left?};
\node[proc,     below=5.0mm of candrem] (take)      {Try next $t$};
\node[decision, below=5.0mm of take]   (cons)      {Consistent?};
\node[proc,     below=5.0mm of cons]   (update)    {Set $\vec{\pi}_i = t$};
\node[decision, below=5.0mm of update] (fcomp)     {$\vec{\pi}$ complete?};
\node[verify,   below=5.0mm of fcomp]  (ver)       {Phase/dependency OK?};
\node[terminal, below=of ver]          (done)      {Return $\underline{\Pi}$};

% ----- backtracking column -----
\node[decision, right=10mm of candrem] (btavail)   {Backtrack?};
\node[proc,     above=7mm of btavail]  (undo)      {Undo last\\assignment};
\node[terminal, below=7mm of btavail]  (nosym)     {No symmetry};

% ----- guide lines -----
% two separate left guides so the two return loops do not overlap
\coordinate (leftInner) at ($(select.west)+(-6.5mm,0)$);
\coordinate (leftOuter) at ($(select.west)+(-12mm,0)$);

% right guide outside the widest right-column box
\coordinate (rightGuide) at ($(undo.east)+(6mm,0)$);

% ----- main flow -----
\draw[arrow] (init) -- (select);
\draw[arrow] (select) -- (build);
\draw[arrow] (build) -- (candrem);

\draw[arrow] (candrem) -- node[lab, pos=0.34] {yes} (take);
\draw[arrow] (take) -- (cons);

\draw[arrow] (cons) -- node[lab, pos=0.34] {yes} (update);
\draw[arrow] (update) -- (fcomp);

\draw[arrow] (fcomp) -- node[lab, pos=0.34] {yes} (ver);
\draw[arrow] (ver) -- node[lab, pos=0.34] {yes} (done);

% ----- left-side return loops -----
\draw[arrow]
  (cons.west) -- node[lab, pos=0.22] {no} (cons.west -| leftInner)
  |- (candrem.west);

\draw[arrow]
  (fcomp.west) -- node[lab, pos=0.22] {no} (fcomp.west -| leftOuter)
  |- (select.west);

% ----- transitions to backtracking -----
\draw[arrow] (candrem.east) -- node[lab, pos=0.45] {no} (btavail.west);

\draw[arrow]
  (ver.east) -- node[lab, pos=0.22] {no} (ver.east -| rightGuide)
  |- (btavail.east);

% ----- backtracking flow -----
\draw[arrow] (btavail) -- node[lab, pos=0.34] {yes} (undo);

\draw[arrow]
  (undo.east) -- (undo.east -| rightGuide)
  |- (select.east);

\draw[arrow] (btavail) -- node[lab, pos=0.34] {no} (nosym);

\end{tikzpicture}
\caption{Search algorithm for constructing a permutation $\underline{\Pi}$.}
\label{fig:algorithm}
\end{figure}
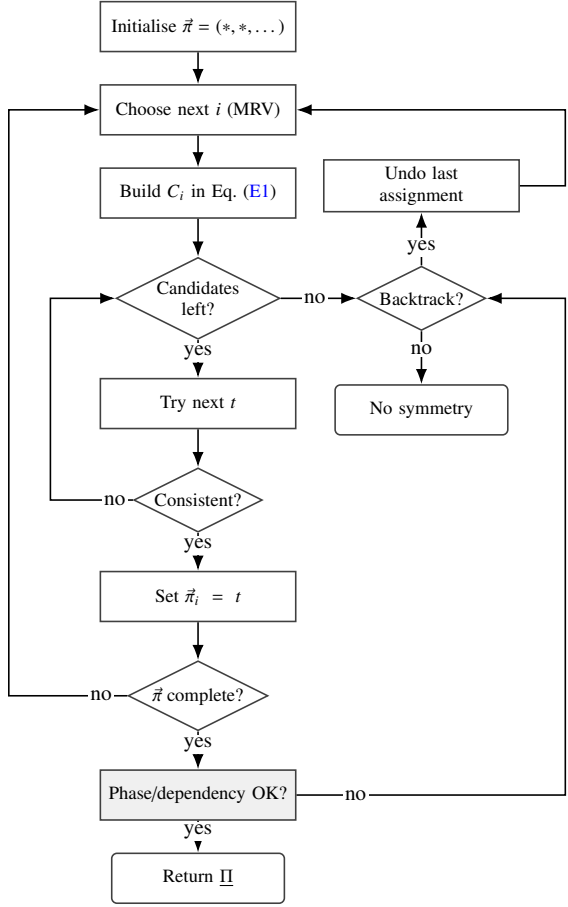

The essential element of the search algorithm is to construct a permutation $\underline{\Pi}$ from the corresponding vector $\vec{\pi}$. 
As shown in Fig.~\ref{fig:algorithm}, we initialize $\vec{\pi} = (*, *, \dots)$, where $*$ denotes an unassigned entry $\pi_i$. We then iteratively assign candidate values to $\pi_i$ under the necessary restrictions imposed by graph colouring, until a complete possible permutation is obtained. 
%We emphasise that the scheme presented is one possible choice for the search for $\vec{\pi}$, and one can imagine different schemes to obtain the permutation, which may be optimised for some set of Hamiltonians.
%
%Each element $\pi_i$ denotes the Pauli index to which the $i$th Pauli is mapped - so $\pi_1 = 2$ means $\vec{p}_1 \to \vec{p}_2$ under the permutation $\underline{\Pi}$. 
At each step of the algorithm we choose an unmapped label $i$ by a minimum-remaining-values (MRV) \cite{HARALICK1980263} heuristic (see below) measured against the base colour sets. This essentially means we choose the $i$ with the fewest available options with the same colour that is not already in the set of targets $\vec{\pi}$. We thus keep a list $C_i$ of feasible targets for each vertex $i$, written as 
\begin{equation}\label{eq:colour_classes}
    C_i = \{t \, | \, c_t = c_i\; {\rm and} \; t \notin \vec{\pi}\},
\end{equation}
where $c_i$ includes the coefficient of Pauli $\vec{p}_i$ as well as other labels included in the graph preprocessing (see Secs.~\ref{app:building_graph}, \ref{app:WL} and \ref{app:heuristic}). 

A proposal $i\mapsto t$ is accepted into $\vec{\pi}$ if it satisfies local symplectic consistency with all already mapped vertices:
\begin{equation}
\label{eq:local-consistency}
\langle \vec{p}_i, \,  \vec{p}_j \rangle = \langle \vec{p}_t,\, \vec{p}_{\pi_j}\rangle
\;\text{and}\
\langle \vec{p}_j, \, \vec{p}_i\rangle = \langle \vec{p}_{\pi_j},\, \vec{p}_t\rangle,
\end{equation}
for each $j$ with already assigned $\pi_j$ \footnote{
For $d=2$ one has $\langle \vec{p}_i, \, \vec{p}_j\rangle=\langle \vec{p}_j, \, \vec{p}_i\rangle$, so the second equality is redundant.
}. If this test passes, we extend $\vec{\pi}$ and continue; otherwise we backtrack. When $\vec{\pi}$ is complete (no `$*$' remains) we perform the final verification tests (see below and Secs.~\ref{app:linear_codes} and \ref{app:phase_correction}).

On the other hand, if the test of the symplectic conditions in Eq.~\eqref{eq:local-consistency} fails for the current choice $i \mapsto t$, that choice is rejected (see Fig.~\ref{fig:algorithm}) and the search tries the next admissible $t$ for the same $i$. If no admissible $t$ remains for $i$, the algorithm backtracks one level, removing the most recent assignment, restoring $\vec{\pi}$ and the available target set, and attempts a different target for the previous vertex. This process repeats (possibly backtracking multiple levels) until a new admissible branch is found or the search space is exhausted.

The runtime of the backtracking search can vary substantially depending on (i) the order in which unmapped vertices $i$ are selected, and (ii) the order in which admissible targets $t \in C_i$ are tried. We therefore allow for a standard deterministic ordering as well as a randomized ordering used for parallelised search. 
In the standard scheme, the next vertex $i$ (see Fig.~\ref{fig:algorithm}) is chosen based on the current base partition. Among unmapped $i$, the MRV subroutine prefers one whose colour set $C_i$ has the fewest remaining unused images. Furthermore, the candidates $t\in C_i$ are enumerated in a deterministic order.

This candidate ordering is not simply the raw vertex index order. 
Instead, candidates are sorted using a precomputed ordering key derived from the colours together with symplectic-product information. 
This ordering affects the order in which branches are explored, as those that are less likely to lead to further branching are preferred. It does not change the hard feasibility constraints nor the correctness of the search.

For some structured Hamiltonians, the MRV score can exhibit frequent ties. In such cases, the search can be highly sensitive to the ordering heuristic, with different branch orders leading to a large variation in search depth before a symmetry is found and hence to large runtime variance. To avoid this, we implement randomized search ordering that is varied between runs.
Since our goal in the present work is typically to find a single symmetry (rather than enumerate the full automorphism group), such randomized starts parallelize naturally.

Finally, whenever a mapping $\vec{\pi}$ is fully constructed (no `$*$' left) we build its permutation matrix $\underline{\Pi}$, and obtain a candidate $\underline{S}$ from Eq.~\eqref{eq:basis_map}. We then test the permutation code automorphism condition, Eq.~\eqref{eq:code_auto_intro}, and ensure phase correction (Sec.~\ref{app:phase_correction}). If either of these steps fails, we backtrack the search.

The algorithm, summarized in Fig.~\ref{fig:algorithm}, can be illustrated via the four-vertex example introduced above, with vertices
$V=\{1,2,3,4\}$ and coefficient colours such that only $1$ and $2$ may be exchanged, i.e. $c_1 = c_2$ and $c_3 \neq c_4 \neq c_1$. As such, in the first two iterations of our algorithm, the MRV chooses indices $i=3$ and $i=4$ to yield
\begin{equation*}
    \vec{\pi} = \{*, *, 3,\ 4 \}.
\end{equation*}
In fact, these elements are both characterized by a list of feasible targets [Eq.~\eqref{eq:colour_classes}] that contains one element (themselves: having distinct colours, they cannot be mapped to other elements). Thus the MRV heuristic trivially sets these permutations first.
The remaining vertices are characterized by the same colour, and thus (aside from self-mapping) we can have either of the following
\begin{equation*}
1 \mapsto 2, \qquad 2 \mapsto 1,
\end{equation*}
as next index $i$ to be chosen by the algorithm. In these cases, tests from Eq.~\eqref{eq:colour_classes} pass, finally yielding (after two iterations) $\vec{\pi} = \{2,1,3,4\}$. This candidate $\vec{\pi}$ is further tested for automorphisms and phases as necessary.

To implement the MRV we do not recompute at every step.
The implementation maintains, for each colour label $\tilde c$, the number of currently unused target vertices of that colour,
\[
    R(\tilde c)=\#\{j:\; j \text{ is unused and } \tilde c_j=\tilde c\}.
\end{equation*}
The next vertex is selected by minimizing $R(\tilde c_i)$ over unmapped $i$. Thus the MRV score is computed from the remaining size of the relevant colour bucket and updated incrementally as assignments are made and undone.

This is an MRV heuristic rather than the exact size of the locally filtered candidate set, since local symplectic consistency [Eq.~\eqref{eq:local-consistency}] is checked only when candidate images are tested.

Among vertices with equal MRV score, ties are broken by a fixed domain order built from the base colour classes. Concretely, the base colour classes are first ordered by decreasing class size, and the vertices within each class are then traversed in a stored order. This order is randomised for parallel searches. Thus different seeds explore the same constrained search tree with different branch-order priorities.

\subsection{Additional graph colouring from heuristic}\label{app:heuristic}

%\Luca{can't we use matroid to provide extra colouring that are truly invariant under the action of clifford?}
%\Charlie{Yes, I have a code that does it, but it doesn't speed up anything in the testing I've done with it}

To lower the number of candidate permutations $\vec{\pi}$, it is possible to add additional colours to the graph prior to WL refinement described in Sec.~\ref{app:WL}. The idea is to employ a heuristic that changes the graph colours based on the Pauli dependencies encoded in the generator matrix $\underline{G}$, see Sec.~\ref{app:linear_codes}. Since the $i$-th column of $\underline{G}$ records how Pauli $i$ is represented in the chosen generating set, column statistics can be used to enrich vertex colour.

Concretely, for each vertex $i$, we define the associated histogram $\mathrm{hist}(i)$ where the $a$-th column has height
\begin{equation}
    h_i(a) = \#\{\mu\in\{1,\dots,n_b\}| \underline{G}_{\mu i}=a\},
    \qquad a\in\mathbb{Z}_d.
\end{equation}
Here, `$\#$' counts elements in the set. 
%For qubits ($d=2$), this reduces (up to the trivial relation $h_i(0)=n_b-h_i(1)$) to the Hamming weight (number of non-zero elements) of the column. In effect, 
%$\mathrm{hist}(i)$ is a crude connectivity measure of Pauli $i$ to the chosen dependence relations.
Once $\mathrm{hist}(i)$ is computed for all Pauli strings, we perform a colour refinement such that, rather than the coefficient $c_i$ alone, each vertex is characterized by the pair $\{ c_i,\mathrm{hist}(i) \}$. 

%We then augment the initial WL seed by including this column profile,
%\begin{equation}
%    \kappa_i^{(0)} = \big(c_i,\mathrm{hist}(i),\ldots\big),
%\end{equation}
%and run the same refinement procedure as in Sec.~\ref{app:WL}.

The idea of this additional graph colouring is to add a feature that is invariant under the action of Clifford gates. However, $\mathrm{hist}(i)$ is not truly invariant, as it depends on the chosen $\underline{G}=[\underline{\mathbb{1}}_{n_b} \mid \underline{D}]$. Performing the change of basis in Eq.~\eqref{eq:U_row_intro} could change $\mathrm{hist}(i)$, leading to failure in finding available symmetries. Nevertheless, as discussed in Sec.~\ref{app:numerical_results}, this heuristic greatly speeds up symmetry finding as it can substantially reduce branching. One could similarly introduce such a heuristic via graph augmentation, by introducing only a subset of the minimal circuits, corresponding to a description of the dependencies in a particular basis.

\section{Example Hamiltonians}\label{app:hamiltonians}

\subsection{Random Pauli Hamiltonians}
In order for a thorough benchmark of the algorithm, we generate random Hamiltonians which we may guarantee have a Clifford symmetry of a given form.
A convenient way to generate random Hamiltonians with an injected Clifford symmetry is to construct them as sums of orbit-averaged Pauli seeds. Let a Clifford symmetry act on Pauli rows by $\vec p\mapsto \vec p\,\underline F$, with phase updates tracked by Eq.~\eqref{eq:phase_update}. For a randomly generated seed Pauli $\hat P(\vec p_k)$, define its orbit under repeated action of the symmetry and form the symmetrised block
\begin{equation}
    \hat H_k^{(\mathrm{sym})}
    =
    \sum_{t=0}^{L_k-1}
    c_k
    e^{
    \frac{
    i \pi
    }{
    d
    }
    \eta_{k,t}
    }
    \hat{P}
    \left(
    \vec p_k \underline F^{\,t}
    \right)
    ,
\end{equation}
where $L_k$ is the orbit length and $\eta_{k,t}\in\mathbb Z_{2d}$ are obtained by iterating the phase-update rule. By construction, each block is invariant under the Clifford action (up to a relabelling of terms within the orbit), and hence any linear combination $\hat H=\sum_k \hat H_k^{(\mathrm{sym})}$ has the injected symmetry. We note that the number of Paulis in this approach is not precisely controlled, as the size of the orbit is not known a priori for a given seed Pauli.

\subsection{All-to-all Heisenberg model}\label{app:heisenberg_model}

We consider the Hamiltonian
\begin{equation}\label{eq:app_heisenberg_hamiltonian}
    \hat{H}_{XXZ}
    =
    \sum_{1\le i<j\le n} J_{ij}
    \left(
        \hat X_i \hat X_j
        +
        \hat Y_i \hat Y_j
        +
        \Delta \hat Z_i \hat Z_j
    \right)
    +
    \sum_{i=1}^{n} h_i \hat Z_i .
\end{equation}
Here $J_{ij} = J_{XXZ}$ are pairwise couplings on the complete graph and $h_i$ are local longitudinal fields. These are selected from a random uniform distribution of unit variance. $\Delta$ is an anisotropy parameter, $\Delta = 1$ yields the XXX model, we choose $\Delta = 0.5$.

\subsection{Transverse-Field Ising Model}

\subsubsection{Closed-system model}

We consider a system of $n$ qubits with transverse-field Ising Hamiltonian
\begin{equation}\label{eq:app_ising_hamiltonian}
    \hat{H}_{\rm TFI}
    =
    J_{\rm TFI}\sum_{\langle i, j \rangle}^{n-1} \hat Z_i \hat Z_{j}
    +
    h_{\rm TFI}\sum_{i=1}^{n}  \hat X_i,
\end{equation}
where $\hat X_i,\hat Z_i$ are Pauli operators acting on site $i$, $J_{\rm TFI}$ are nearest-neighbour Ising couplings, and $h_{\rm TFI}$ are transverse fields. Here we assume open boundary conditions. The sum over $\langle i, j \rangle$ denotes a sum over nearest neighbours in a given geometry.

\subsubsection{Ising chain with local dephasing}

For the open-system case we use the GKSL master equation of Eq.~\eqref{eq:lindblad_master_eq}
with $\hat H$ given by Eq.~\eqref{eq:app_ising_hamiltonian} and local dephasing jump operators
\begin{equation}\label{eq:app_dephasing_jumps}
    \hat L_i = \sqrt{\gamma_i}\,\hat Z_i,
\end{equation}
where $\gamma_i\ge 0$ are site-dependent dephasing rates (or $\gamma_i=\gamma_z$ for uniform dephasing).

Since $\hat Z_i^\dagger \hat Z_i = \hat{\mathbb I}$, each dissipative contribution simplifies to
\begin{equation}\label{eq:app_dephasing_term}
    \hat L_i \hat\rho \hat L_i^\dagger
    -
    \frac{1}{2}\{\hat L_i^\dagger \hat L_i,\hat\rho\}
    =
    \gamma_i \big(\hat Z_i \hat\rho \hat Z_i - \hat\rho\big),
\end{equation}
so that the Liouvillian is
\begin{equation}\label{eq:app_liouvillian_dephasing}
    \mathcal L(\hat\rho)
    =
    -i[\hat H,\hat\rho]
    +
    \sum_{i=1}^{n}
    \gamma_i \big(\hat Z_i \hat\rho \hat Z_i - \hat\rho\big).
\end{equation}

\subsection{Toric Code}

For the Toric code \cite{BullockBrennen2007QuditSurface, Resende_toric, Chan_VQE_2024, ferguson2021measurement} on a square lattice with qudits on edges, the Hamiltonian is naturally expressed as a Pauli sum in the row-tableau convention of Eq.~\eqref{eq:pauli_sum}, with each term represented by a row $\vec p_i=(\vec x_i,\vec z_i)\in\mathbb F_d^{2n}$. In the qubit case ($d=2$), star operators are purely $Z$-type and plaquette operators that are purely $X$-type in Eq.~\eqref{eq:toricH}. In tableau form, each star contributes a row of the form $(\vec 0,\vec z_s)$ and each plaquette contributes a row of the form $(\vec x_p,\vec 0)$, with support determined by the corresponding set of edges in the lattice indexing used by the code.

For prime $d>2$, the $\mathbb Z_d$ qudit Toric code \cite{BullockBrennen2007QuditSurface} is obtained by replacing qubit Paulis with generalised Weyl operators $\hat X,\hat Z$ and assigning an orientation to edges so that star and plaquette operators carry signed exponents $\pm 1\in\mathbb F_d$ (implemented as $1$ and $d-1$ modulo $d$):
\begin{equation}
    \hat A_s=\prod_e \hat Z_e^{\,\sigma_{s,e}},
    \qquad
    \hat B_p=\prod_e \hat X_e^{\,\tau_{p,e}},
    \qquad
    \sigma_{s,e},\tau_{p,e}\in\{0,\pm 1\}.
\end{equation}
To match the implementation, the default qudit Hamiltonian is taken in Hermitian form
\begin{equation}
    \hat H_{\mathrm{TC}}^{(d)}
    =
    c_z \sum_s \bigl(\hat A_s+\hat A_s^\dagger\bigr)
    +
    c_x \sum_p \bigl(\hat B_p+\hat B_p^\dagger\bigr)
    +
    c_g \sum_e \bigl(\hat Z_e+\hat Z_e^\dagger\bigr),
\end{equation}
where the final term is the single-edge $Z$-type field/gauge term used in the code (with no $X$-field term included). For $d=2$, one has $\hat A_s^\dagger=\hat A_s$ and $\hat B_p^\dagger=\hat B_p$, and the oriented signs reduce modulo $2$, recovering the usual qubit Toric code structure (up to coupling normalization conventions).

\subsection{Spinless fermionic $t$-$V$ chain}

A standard interacting lattice-fermion model in one dimension is the spinless fermionic $t$-$V$ chain (see e.g. Refs.  \cite{Vardhan_fermionic_2017, KieferEmmanouilidis_Unlimited_2021}),
\begin{equation}
\hat{H}_{\mathrm{tV}}
=
- J_{tV} \sum_{\langle i,j\rangle}
\left(\hat{c}_i^\dagger \hat{c}_j + \hat{c}_j^\dagger \hat{c}_i\right)
+ \sum_i D_i\, \hat{n}_i
+ V_{tV} \sum_{\langle i,j\rangle} \hat{n}_i \hat{n}_j,
\end{equation}
Here $J_{tV}$ is the nearest-neighbour hopping amplitude, $D_i$ is a site-dependent on-site potential, and $V_{tV}$ is a nearest-neighbour interaction. The operator $\hat{c}_{i}^{(\dagger)}$ annihilates (creates) a spinless fermion on site $i$, and $\hat{n}_{i}=\hat{c}_{i}^\dagger \hat{c}_{i}$.
We study the case where $D_i$ is a random disorder. The model conserves total particle number $\hat{N}=\sum_i \hat{n}_i$, and provides a minimal setting for studying transport, localization, and interaction effects in one-dimensional fermionic systems. We map the spinless fermionic $t$-$V$ chain to qubits using the standard Jordan-Wigner convention:
\begin{equation*}
\hat{n}_j
=
\hat{c}_j^\dagger \hat{c}_j
=
\frac{\mathbb{1}-\hat{Z}_j}{2}.
\end{equation*}
For an open chain, the Jordan--Wigner transformation gives
\begin{equation*}
\hat{c}_j
=
\left(\prod_{\ell<j}\hat{Z}_\ell\right)
\frac{\hat{X}_j+i\hat{Y}_j}{2},
\qquad
\hat{c}_j^\dagger
=
\left(\prod_{\ell<j}\hat{Z}_\ell\right)
\frac{\hat{X}_j-i\hat{Y}_j}{2}.
\end{equation*}
Substituting these expressions into the spinless fermionic Hamiltonian gives, up to an additive identity shift,
\begin{equation*}
\hat{H}_{\mathrm{tV}}
=
-\frac{J_{tV}}{2}
\sum_{\langle i,j\rangle}
\left(
\hat{X}_i\hat{X}_j
+
\hat{Y}_i\hat{Y}_j
\right)
-\frac{1}{2}
\sum_i D_i \hat{Z}_i
+
\frac{V_{tV}}{4}
\sum_{\langle i,j\rangle}
\left(
\hat{Z}_i\hat{Z}_j
-
\hat{Z}_i
-
\hat{Z}_j
\right).
\end{equation*}
Identity terms are omitted since they do not affect the symmetry structure.

\subsection{Fermi-Hubbard model}

The Fermi-Hubbard model is a canonical model of strongly correlated electrons on a lattice \cite{Hubbard_Electron_1963}:
\begin{equation}
\hat{H}_{\mathrm{FH}}
=
- t \sum_{\langle i,j\rangle,\sigma}
\left(\hat{c}_{i\sigma}^\dagger \hat{c}_{j\sigma} + \hat{c}_{j\sigma}^\dagger \hat{c}_{i\sigma}\right)
+ U \sum_i \hat{n}_{i\uparrow} \hat{n}_{i\downarrow}
- \mu \sum_{i,\sigma} \hat{n}_{i\sigma}.
\end{equation}
Here $t$ is the nearest-neighbour hopping amplitude, $U$ is the on-site interaction strength, and $\mu$ is the chemical potential. The operator $\hat{c}_{i\sigma}^{(\dagger)}$ annihilates (creates) a fermion with spin $\sigma\in\{\uparrow,\downarrow\}$ on site $i$, and $\hat{n}_{i\sigma}=\hat{c}_{i\sigma}^\dagger \hat{c}_{i\sigma}$. This model is mapped to spins via the Jordan-Wigner transformation using OpenFermion \cite{McClean_2020}.

\end{document}